\documentclass[sigplan,screen,10pt]{acmart}

\acmConference[LICS '20]{Proceedings of the 35th Annual ACM/IEEE Symposium on Logic in Computer Science (LICS)}{July 8--11, 2020}{Saarbr\"ucken, Germany}
\acmBooktitle{Proceedings of the 35th Annual ACM/IEEE Symposium on Logic in Computer Science (LICS '20), July 8--11, 2020, Saarbr\"ucken, Germany}
\acmYear{2020} 
\acmISBN{978-1-4503-7104-9/20/07}
\acmDOI{10.1145/3373718.3394743}
\startPage{1}

\setcopyright{rightsretained} 
\copyrightyear{2020}           

\usepackage{subcaption} 
\usepackage[utf8]{inputenc}
\usepackage[T1]{fontenc}
\usepackage{amsmath}
\usepackage{amssymb}
\usepackage{amsthm}
\usepackage{cleveref}
\usepackage{todonotes}
\usepackage{csquotes}
\usepackage{enumitem}
\usepackage{placeins}
\usepackage{thm-restate}
\usepackage{mathtools}

\usetikzlibrary{automata,positioning}

\tikzset{state/.style={
  rectangle,
  rounded corners,
  draw=black,
  minimum height=2em,
  minimum width=2em,
  align=center}
}

\tikzset{every picture/.append style={initial text=}}
\tikzset{accepting/.style = {double}}
\tikzset{>=stealth}
\tikzset{parallel above/.append style={transform canvas={yshift= 1mm}}}
\tikzset{parallel below/.append style={transform canvas={yshift=-1mm}}}
\tikzset{parallel right/.append style={transform canvas={xshift= 1.3mm}}}
\tikzset{parallel left/.append style={transform canvas={xshift=-1.3mm}}}

\newcommand{\lang}[0]{\ensuremath{\mathcal{L}}}

\newcommand{\true}{{\ensuremath{\mathbf{t\hspace{-0.5pt}t}}}}
\newcommand{\false}{{\ensuremath{\mathbf{ff}}}}
\newcommand{\F}{{\ensuremath{\mathbf{F}}}}
\renewcommand{\G}{{\ensuremath{\mathbf{G}}}}
\newcommand{\X}{{\ensuremath{\mathbf{X}}}}
\renewcommand{\U}{{\ensuremath{\mathbf{U}}}}
\newcommand{\W}{{\ensuremath{\mathbf{W}}}}
\newcommand{\M}{{\ensuremath{\mathbf{M}}}}
\newcommand{\R}{{\ensuremath{\mathbf{R}}}}

\newcommand{\subf}{\textit{sf}\,}

\newcommand{\sfmu}{{\ensuremath{\mathbb{\mu}}}}
\newcommand{\sfnu}{{\ensuremath{\mathbb{\nu}}}}
\newcommand{\setmu}{\ensuremath{M}}
\newcommand{\setnu}{\ensuremath{N}}
\newcommand{\setF}{\ensuremath{\mathcal{F}}}
\newcommand{\setG}{\ensuremath{\mathcal{G}}}
\newcommand{\setFG}{\ensuremath{\mathcal{F\hspace{-0.1em}G}}}
\newcommand{\setGF}{\ensuremath{\mathcal{G\hspace{-0.1em}F}\!}}

\newcommand{\evalnu}[2]{{#1[#2]^\Pi_1}}
\newcommand{\evalmu}[2]{{#1[#2]^\Sigma_1}}
\newcommand{\flatten}[2]{{#1[#2]^\Sigma_2}}
\newcommand{\flattentwo}[2]{{#1[#2]^\Pi_2}}

\newcommand{\trans}[1]{\overset{#1}{\longrightarrow}}

\newcommand{\Op}{\mathbin{op}}

\newtheorem{theorem}{Theorem}
\newtheorem{definition}[theorem]{Definition}
\newtheorem{lemma}[theorem]{Lemma}
\newtheorem{corollary}[theorem]{Corollary}
\newtheorem{proposition}[theorem]{Proposition}
\newtheorem{example}[theorem]{Example}
\newtheorem{remark}[theorem]{Remark}

\newcommand{\aww}{\textnormal{AWW}}
\newcommand{\alw}{\textnormal{A1W}}
\newcommand{\dbw}{\textnormal{DBW}}
\newcommand{\dcw}{\textnormal{DCW}}
\newcommand{\drw}{\textnormal{DRW}}

\newcommand{\A}{\textbf{A}}

\newcommand{\SetToWidest}[1]{\mathmakebox[6.5cm][l]{#1}}%
\newcommand{\SetToWidestTwo}[1]{\mathmakebox[0.75cm][r]{#1}}%

\newif\ifarxiv
\arxivtrue

\begin{document}

\title[An Efficient Normalisation Procedure for LTL and A1W]{An Efficient Normalisation Procedure for Linear Temporal Logic and Very Weak Alternating Automata}

\ifarxiv
\subtitle{Extended Version}
\fi


\author{Salomon Sickert}

\orcid{0000-0002-0280-8981}            

\affiliation{
  \institution{Technische Universität München} 
  \country{Germany}                 
}

\email{s.sickert@tum.de}               

\author{Javier Esparza}               
\orcid{0000-0001-9862-4919}

\affiliation{
  \institution{Technische Universität München} 
  \country{Germany}                 
}

\email{esparza@in.tum.de}  

\begin{abstract}
In the mid 80s, Lichtenstein, Pnueli, and Zuck proved a classical theorem stating that every formula of Past LTL (the extension of LTL with past operators) is equivalent to a formula of the form $\bigwedge_{i=1}^n \G\F \varphi_i \vee \F\G \psi_i $, where $\varphi_i$ and $\psi_i$ contain only past operators. Some years later, Chang, Manna, and Pnueli built on this result to derive a similar normal form for LTL. Both normalisation procedures have a non-elementary worst-case blow-up, and follow an involved path from formulas to counter-free automata to star-free regular expressions and back to formulas. We improve on both points. We present a direct and purely syntactic normalisation procedure for LTL yielding a normal form, comparable to the one by Chang, Manna, and Pnueli, that has only a single exponential blow-up. As an application, we derive a simple algorithm to translate LTL into deterministic Rabin automata. The algorithm normalises the formula, translates it into a special very weak alternating automaton, and applies a simple determinisation procedure, valid only for these special automata. 
\end{abstract}

\ccsdesc[500]{Theory of computation~Modal and temporal logics}
\ccsdesc[500]{Theory of computation~Automata over infinite objects}

\keywords{Linear Temporal Logic, Normal Form, Weak Alternating Automata, Deterministic Automata}  

\maketitle

\section{Introduction}

In seminal work carried out in the middle 80s, Lichtenstein, Pnueli, and Zuck 
investigated Past Linear Temporal Logic (Past LTL), a temporal logic with future and past operators. They proved the classical result stating that every formula is equivalent to another one
of the form 
\begin{equation}
\label{eq:nf}
\bigwedge_{i=1}^n \G\F \varphi_i \vee \F\G \psi_i 
\end{equation}
\noindent where $\varphi_i$ and $\psi_i$ only contain past operators  \cite{DBLP:conf/lop/LichtensteinPZ85,XXXX:phd/Zuck86}. Shortly after, Manna and Pnueli
introduced the \emph{safety-progress} hierarchy, containing six classes of properties (\Cref{fig:temporal_hierarchy}), and presented a logical characterisation of each class in terms of syntactic fragments of Past LTL \cite{DBLP:conf/podc/MannaP89,DBLP:books/daglib/0077033}. The class of \emph{reactivity properties}, placed at the top of the hierarchy, contains all Past LTL properties, and its syntactic characterisation, given by (\ref{eq:nf}), is the class of \emph{reactivity formulas}.

In the early 90s, LTL (which only has future operators, but is known to be as expressive as Past LTL), became the logic of choice for most model-checking applications. At that time Chang, Manna, and Pnueli showed that the classes of the safety-progress hierarchy also admit syntactic characterisations in terms of LTL fragments \cite{DBLP:conf/icalp/ChangMP92}. In particular,  
they proved that every LTL formula is equivalent to another one in which every path of the syntax tree alternates at most once between the \enquote{least-fixed-point} operators $\U$ and $\M$ and 
the \enquote{greatest-fixed-point} operators $\W$ and $\R$. In the notation introduced in \cite{DBLP:conf/mfcs/CernaP03}, which mimics the definition of the $\Sigma_i$, $\Pi_i$, and $\Delta_i$ classes of the arithmetical and polynomial hierarchies, they proved that every LTL formula is equivalent to a $\Delta_2$-formula.

While these normal forms have had large conceptual impact in model checking, automatic synthesis, and deductive verification (see e.g. \cite{PitermanP18} for a recent survey), the  normalisation \emph{procedures} proving that they are indeed normal forms have had none. In particular, contrary to the case of propositional or first-order logic, they have not been implemented in tools. The reason is that they are not direct, have high complexity, and their correctness proofs are involved.  Let us elaborate on this. In \cite{XXXX:phd/Zuck86}, Zuck gives a detailed description of the normalisation procedure of \cite{DBLP:conf/lop/LichtensteinPZ85}. First, Zuck translates the initial Past LTL formula into a counter-free semi-automaton, then applies the Krohn-Rhodes decomposition and other results to translate the automaton into a star-free regular expression, and finally translates this expression into a reactivity formula with a non-elementary blow-up. In \cite{DBLP:conf/podc/MannaP89,DBLP:books/daglib/0077033} the procedure is not even presented, the reader is referred to \cite{XXXX:phd/Zuck86} and/or to previous results\footnote{Including papers by Burgess, McNaughton and Pappet, Choueka, Thomas, and Gabby, Pnueli, Shela, and Stavi.}. The normalisation procedure of \cite{DBLP:conf/icalp/ChangMP92} for LTL calls the translation procedure of \cite{DBLP:conf/lop/LichtensteinPZ85,XXXX:phd/Zuck86} for Past LTL as a subroutine, and so it is not any simpler\footnote{Further, \cite{DBLP:conf/icalp/ChangMP92} only contains a short sketch of the translation of reactivity formulas into $\Delta_2$-formulas.}. Finally, while  
Maler and Pnueli present in \cite{DBLP:conf/focs/MalerP90} an improved translation of star-free regular languages to Past LTL, their work still leads to a triple exponential normalisation procedure for Past LTL. Further, it is not clear to us if this translation can also be used to obtain $\Delta_2$-formulas.

In this paper we present a novel normalisation procedure that translates any LTL formula into an equivalent $\Delta_2$-formula. Our procedure is:
\begin{itemize}
\item \emph{Direct}. It does not require any detour through automa\-ta or regular expressions.
\item \emph{Syntax-guided}. It consists of a few syntactic rewrite rules---not unlike the rules for putting a boolean formula in conjunctive or disjunctive normal form---that can be described in less than a page.
\item \emph{Single exponential}. The length of the $\Delta_2$-formula is at most exponential in the length of the original formula, a dramatic improvement on the previous non-elemen\-tary and triple exponential bounds.
\end{itemize}
The correctness proof of the procedure consists of a few lemmas, all of them with routine proofs by structural induction. 
It is presented in \Cref{sec:overview,sec:eval,sec:newnf}, modulo the omission of some straightforward induction cases. To make this paper self-contained, the proofs of three lemmas taken from \cite{DBLP:conf/lics/EsparzaKS18,DBLP:phd/dnb/Sickert19} are reproduced in 
\ifarxiv
\Cref{sec:updated_proofs}.
\fi
We have mechanised the complete correctness proof in Isabelle/HOL \cite{DBLP:books/sp/NipkowPW02}, building upon previous work \cite{DBLP:conf/itp/0001SS19,DBLP:journals/afp/SeidlS19,DBLP:journals/afp/Sickert16}. 
The formalised proof consists of roughly 1000 lines, from which one can extract a formally verified normalisation procedure consisting of ca. 200 lines of Standard ML code, excluding standard definitions added by the code-generation. 
Both the formalisation and instructions for extracting code are located in \cite{XXXX:journals/afp/Sickert20}.

In the second part of the paper (\Cref{sec:LTLtoDRW,sec:hierarchy}) we use the new normalisation procedure to derive a simple translation of LTL into deterministic Rabin automata (\drw). First, we show that every formula of $\Delta_2$ can be translated into a very weak alternating Büchi automaton (\alw) in which every path has at most one alternation between accepting and non-accepting states. Further, we provide a simple determinisation procedure for these automata, based on a breakpoint construction. The LTL-to-\drw{} translation normalises the formula, transforms it into an \alw{} with at most one alternation, and determinises this intermediate automaton.

Due to space constraints we do not provide an overview of LTL-to-DRW translations and refer the reader to \cite[Ch. 1]{DBLP:phd/dnb/Sickert19}. Furthermore, we only provide a preliminary experimental evaluation of the proposed translations and leave a detailed analysis as future work.

\section{Preliminaries}

Let $\Sigma$ be a finite alphabet. A \emph{word} $w$ over $\Sigma$ is an infinite sequence of letters $a_0 a_1 a_2 \dots$ with $a_i \in \Sigma$ for all $i \geq 0$, and a language is a set of words. A \emph{finite word} is a finite sequence of letters. As usual, the set of all words (finite words) is denoted $\Sigma^\omega$ ($\Sigma^*$). We let $w[i]$ (starting at $i=0$) denote the $i$-th letter of a word $w$. The finite infix $w[i]w[i+1]\dots w[j - 1]$ is abbreviated with $w_{ij}$ and the infinite suffix $w[i] w[i+1] \dots$ with $w_{i}$. We denote the infinite repetition of a finite word $\sigma_1 \dots \sigma_n$ by $(\sigma_1 \dots \sigma_n)^\omega = \sigma_1 \dots \sigma_n \sigma_1 \dots \sigma_n \sigma_1 \dots$.

\begin{definition}[LTL syntax]
\label{def:ltlsyntax}
LTL formulas over a set $Ap$ of atomic propositions are constructed by the following syntax:
\begin{align*}
\varphi \Coloneqq \; & \true \mid \false \mid a \mid \neg a \mid \varphi \wedge \varphi \mid \varphi\vee\varphi \\ 
                     & \mid \X\varphi \mid \varphi\U\varphi \mid \varphi\W\varphi \mid \varphi\R\varphi \mid \varphi\M\varphi 
\end{align*}
\noindent where $a \in Ap$ is an atomic proposition and $\X$, $\U$, $\W$, $\R$, and $\M$ 
are the next, (strong) until, weak until, (weak) release, and strong release operators, respectively. 
\end{definition}

The inclusion of both the strong and weak until operators as well as the negation normal form are essential to our approach. The operators $\R$ and $\M$, however, are added to ensure that every formula of length $n$ in the standard syntax, with negation but only the until operator, is equivalent to a formula of length $O(n)$ in our syntax. They could be removed, if we accept an exponential blow-up incurred by expressing $\R$ with $\W$. The semantics is defined as usual:

\begin{definition}[LTL semantics]
\label{def:ltlsemantics}
Let $w$ be a word over the alphabet $2^{Ap}$ and let $\varphi$ be a formula. The satisfaction relation $w \models \varphi$ is inductively defined as follows:
{\arraycolsep=1.8pt%
\[\begin{array}[t]{lcl}
w \models \true               & &                                                              \\
w \not \models \false         & &                                                              \\
w \models a & \mbox{ iff }    & a \in w[0]                                                     \\
w \models \neg a              & \mbox{ iff } & a \notin w[0]                                   \\
w \models \varphi \wedge \psi & \mbox{ iff } & w \models \varphi \text{ and } w \models \psi   \\
w \models \varphi \vee \psi   & \mbox{ iff } & w \models \varphi \text{ or } w \models \psi    \\
w \models \X \varphi      & \mbox{ iff } & w_1 \models \varphi \\
w \models \varphi \U \psi & \mbox{ iff } & \exists k. \, w_k \models \psi \text{ and } \forall j < k. \, w_j \models \varphi \\
w \models \varphi \M \psi & \mbox{ iff } & \exists k. \, w_k \models \varphi \text{ and } \forall j \leq k. \, w_j \models \psi \\
w \models \varphi \R \psi & \mbox{ iff } & \forall k. \, w_k \models \psi \text{ or } w \models \varphi\M \psi \\
w \models \varphi \W \psi & \mbox{ iff } & \forall k. \, w_k \models \varphi \text{ or } w \models \varphi\U \psi
\end{array}\]}%
We let $\lang(\varphi) \coloneqq \{ w \in (2^{Ap})^\omega : w \models \varphi\}$ denote the language of $\varphi$.
We overload the definition of $\models$ and write $\varphi \models \psi$ as a shorthand for $\lang(\varphi) \subseteq \lang(\psi)$.
\end{definition}

We use the standard abbreviations $\F \varphi \coloneqq \true \, \U \, \varphi$ (eventually) and $\G \varphi \coloneqq \false \, \R \, \varphi$ (always). Finally, we introduce the notion of equivalence of formulas, and equivalence within a language.
\newcommand{\equ}[1]{\equiv^{#1}}

\begin{definition}
Two formulas $\varphi$ and $\psi$ are \emph{equivalent}, denoted $\varphi \equiv \psi$, if $\lang(\varphi) = \lang(\psi)$. Given a language $L \subseteq (2^{Ap})^\omega$, two formulas $\varphi$ and $\psi$ are \emph{equivalent within $L$}, denoted $\varphi \equ{L} \psi$,  if $\lang(\varphi) \cap L = \lang(\psi) \cap L$.
\end{definition}

\section{The Safety-Progress Hierarchy}
\label{sec:hierarchy}

We recall the hierarchy of temporal properties studied by Manna and Pnueli \cite{DBLP:conf/podc/MannaP89} following the formulation of {\v{C}}ern{\'{a}} and Pel{\'{a}}nek \cite{DBLP:conf/mfcs/CernaP03}. In the ensuing sections we describe structures that have a direct correspondence to this hierarchy and in this sense the hierarchy provides a map to navigate the results of this paper.

\begin{definition}[\cite{DBLP:conf/podc/MannaP89,DBLP:conf/mfcs/CernaP03}]
Let $P \subseteq \Sigma^\omega$ be a property over $\Sigma$.
\begin{itemize}
\item $P$ is a safety property if there exists a language of finite words $L \subseteq \Sigma^*$ such that for every $w \in P$ all finite prefixes of $w$ belong to $L$.
\item $P$ is a guarantee property if there exists a language of finite words $L \subseteq \Sigma^*$ such that for every $w \in P$ there exists a finite prefix of $w$ which belongs to $L$.
\item $P$ is an obligation property if it can be expressed as a positive boolean combination of safety and guarantee properties.
\item $P$ is a recurrence property if there exists a language of finite words $L \subseteq \Sigma^*$ such that for every $w \in P$ infinitely many prefixes of $w$ belong to $L$.
\item $P$ is a persistence property if there exists a language of finite words $L \subseteq \Sigma^*$ such that for every $w \in P$ all but finitely many prefixes of $w$ belong to $L$.
\item $P$ is a reactivity property if $P$ can be expressed as a positive boolean combination of recurrence and persistence properties.
\end{itemize}
\end{definition}

The inclusions between these classes are shown in \Cref{fig:temporal_hierarchy}. 
Chang, Manna, and Pnueli give in \cite{DBLP:conf/icalp/ChangMP92} a syntactic characterisation of the classes of the safety-progress hierarchy in terms of fragments of LTL. The following is a corollary of the proof of \cite[Thm. 8]{DBLP:conf/icalp/ChangMP92}:

\begin{definition}[Adapted from \cite{DBLP:conf/mfcs/CernaP03}]
\label{def:future_hierarchy}
We define the following classes of LTL formulas:
\begin{itemize}
	\item The class $\Sigma_0 = \Pi_0 = \Delta_0$ is the least set containing all atomic propositions and their negations, and is closed under the application of conjunction and disjunction.
	\item The class $\Sigma_{i+1}$ is the least set containing $\Pi_i$ and is closed under the application of conjunction, disjunction, and the $\X$, $\U$, and $\M$ operators.
	\item The class $\Pi_{i+1}$ is the least set containing $\Sigma_i$ and is closed under the application of conjunction, disjunction, and the $\X$, $\R$, and $\W$ operators.
	\item The class $\Delta_{i+1}$ is the least set containing $\Sigma_{i+1}$ and $\Pi_{i+1}$ and is closed under the application of conjunction and disjunction.
\end{itemize}
\end{definition}

\begin{theorem}[Adapted from \cite{DBLP:conf/mfcs/CernaP03}]\label{thm:hierarchy:correspondence}
A property that is specifiable in LTL is a guarantee (safety, obligation, persistence, recurrence, reactivity, respectively) property if and only if it is specifiable by a formula from the class $\Sigma_1$, $(\Pi_1$, $\Delta_1$, $\Sigma_2$, $\Pi_2$, $\Delta_2$, respectively$).$
\end{theorem}

\begin{figure}
\begin{subfigure}[c]{0.51\columnwidth}
  \begin{center}
  \small
	\begin{tikzpicture}[x=1cm,y=0.75cm,outer sep=2pt]

    \node (padding1) at ( 0,0.3) {};
	\node (1) at ( 0, 0) {reactivity};
	\node (2) at ( 1,-1) {recurrence};
	\node (3) at (-1,-1) {persistence};
    \node (4) at ( 0,-2) {obligation};
	\node (5) at ( 1,-3) {safety};
	\node (6) at (-1,-3) {guarantee};
	\node (padding2) at ( 0,-3.2) {};

	\path
	(2) edge[draw=none] node[sloped]{$\supset$} (1)
    (3) edge[draw=none] node[sloped]{$\subset$} (1)
    
    (4) edge[draw=none] node[sloped]{$\subset$} (2)
    (4) edge[draw=none] node[sloped]{$\supset$} (3)
    
    (5) edge[draw=none] node[sloped]{$\supset$} (4)
    (6) edge[draw=none] node[sloped]{$\subset$} (4);

	\end{tikzpicture}
  \end{center}
\subcaption{Safety-progress hierarchy \cite{DBLP:conf/podc/MannaP89}}
\label{fig:temporal_hierarchy}
\end{subfigure}%
\begin{subfigure}[c]{0.51\columnwidth}
  \begin{center}
  \small
	\begin{tikzpicture}[x=1cm,y=0.75cm,outer sep=2pt]

    \node (padding1) at ( 0,0.3) {};
	\node (1) at ( 0, 0) {$\Delta_2$};
	\node (2) at ( 1,-1) {$\Pi_2$};
	\node (3) at (-1,-1) {$\Sigma_2$};
    \node (4) at ( 0,-2) {$\Delta_1$};
	\node (5) at ( 1,-3) {$\Pi_1$};
	\node (6) at (-1,-3) {$\Sigma_1$};
    \node (padding2) at ( 0,-3.2) {};

	\path[->]
	(2) edge node{} (1)
    (3) edge node{} (1)
    
    (4) edge node{} (2)
    (4) edge node{} (3)
    
    (5) edge node{} (4)
    (6) edge node{} (4);

	\end{tikzpicture}
  \end{center}
\subcaption{Syntactic-future hierarchy}
\label{fig:syntactic-future}
\end{subfigure}
\caption{Both hierarchies, side-by-side, indicating the correspondence of \Cref{thm:hierarchy:correspondence}}
\label{fig:hierarchies}
\end{figure}
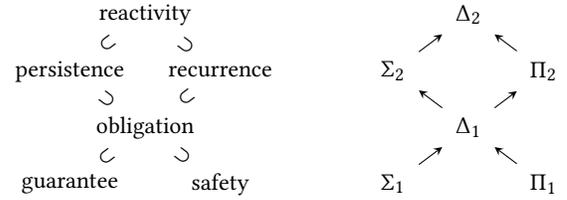

\section{Overview of the Normalisation Result}
\label{sec:overview}

\newcommand{\stablewords}{\mathcal{S}}
\newcommand{\partition}{\mathcal{P}}
\newcommand{\Univ}{\mathcal{U}}

Fix an LTL formula $\varphi$ over a set of atomic propositions $Ap$. Our new normal form is based on two notions: 
\begin{itemize}
\item A partition of the universe $\Univ\coloneqq (2^{Ap})^\omega$ of all words into equivalence classes of words that, loosely speaking, exhibit the same \enquote{limit-behaviour} with respect to $\varphi$.
\item The notion of stable word with respect to $\varphi$.
\end{itemize}

\paragraph{A partition of $\Univ$.}  Let $\sfmu(\varphi)$ and $\sfnu(\varphi)$ be the sets containing the subformulas of $\varphi$ of the form $\psi_1 \Op \psi_2$  for $\Op \in \{\U, \M\}$ and $\Op \in \{\W, \R\}$, respectively. Given a word $w$, define:
\begin{align*}
\setGF_w^{\;\varphi} & \coloneqq \{\psi \colon \psi \in \sfmu(\varphi) \wedge w \models \G\F\psi \} \\
\setFG_w^\varphi     & \coloneqq \{\psi \colon \psi \in \sfnu(\varphi) \wedge w \models \F\G\psi \}
\end{align*}
(To simplify the notation, when $\varphi$ is clear from the context we simply write $\setGF_w$ and $\setFG_w$.) Two words $w, v$ have the same \emph{limit-behaviour} w.r.t. $\varphi$ if $\setGF_w=\setGF_{v}$ and $\setFG_w=\setFG_{v}$. Having the same limit-behaviour is an equivalence relation, which induces the partition $\mathcal{P} = \{ \partition_{\setmu, \setnu} \subseteq \Univ \colon \setmu \subseteq \sfmu(\varphi), \setnu \subseteq \sfnu(\varphi)\}$ given by:
\begin{equation}
\label{eq:partition}
\partition_{\setmu, \setnu} \coloneqq \{ w \in \Univ \colon \setmu = \setGF_w \; \wedge \; \setnu = \setFG_w \} 
\end{equation} 

\begin{example}
Let $\varphi = \G a \vee b \U c$. We have $\sfmu(\varphi) = \{ b \U c\}$ and 
$\sfnu(\varphi) = \{ \G a\}$. The partition $\mathcal{P}$ has four equivalence classes: 
\begin{itemize}
\item $\partition_{\emptyset, \emptyset}$ contains all words such that $b \U c$ holds only finitely often and $\G a$ fails infinitely often (which in this case implies that $\G a$ never holds), e.g. $\{b\}^\omega$ or $\{c\}\{b\}^\omega$.
\item $\partition_{\emptyset, \{\G a\}}$ contains all words such that $b \U c$ holds finitely often and $\G a$ fails finitely often, e.g. $\{a\}^\omega$ or $\{c\}\{a\}^\omega$.
\item $\partition_{\{b \U c\}, \emptyset}$  contains all words such that $b \U c$ holds infinitely often and $\G a$ fails infinitely often, e.g. $(\{a\}\{c\})^\omega$ or $\{a\} \{c\}^\omega$.
\item $\partition_{\{b \U c\}, \{\G a\}}$ contains all words such that  $b \U c$ holds infinitely often and $\G a$ fails finitely often, e.g. $\{b \} \{a,c\}^\omega$ or $(\{a,c\} \{a\})^\omega$.
\end{itemize}
The partition is graphically shown in Figure \ref{fig:ex:partition}. The equivalence classes are shown in blue, red, yellow, and green (ignore the inner part in darker colour for the moment).
\end{example}

\definecolor{lightgreen}{RGB}{152,179,155}
\definecolor{lightblue}{RGB}{56,165,215}
\definecolor{lightred}{RGB}{225,130,133}
\definecolor{lightyellow}{RGB}{254,226,170}

\begin{figure}[btp]
  \begin{center}
  \small
\begin{tikzpicture}[xscale=4,yscale=3]

\filldraw[lightred, opacity=0.80]
  (0,0) rectangle (1,0.5);
\filldraw[lightgreen, opacity=0.80]
  (2,0) rectangle (1,0.5);
\filldraw[lightblue, opacity=0.80]
  (0,1) rectangle (1,0.5);
\filldraw[lightyellow, opacity=0.80]
  (2,1) rectangle (1,0.5);

\filldraw[black, opacity=0.3] 
  (0.5,0.25) rectangle (1.5,0.75);

\draw[thick]
  (0,0) rectangle (2,1);
\draw[]
  (0.5,0.25) rectangle (1.5,0.75);

\draw[]
  (1,0) -- (1,1);
\draw[]
  (0.0,0.5) -- (2,0.5);
 
\draw
  (0.25,0.625) node {$\{c\}\{b\}^\omega$};
\draw
  (0.75,0.625) node {$\{b\}^\omega$};

\draw
  (1.25,0.625) node {$(\{a\}\{c\})^\omega$};
\draw
  (1.75,0.625) node {$\{a\}\{c\}^\omega$};
 
\draw
  (0.25,0.375) node {$\{c\}\{a\}^\omega$};
\draw
  (0.75,0.375) node {$\{a\}^\omega$};

\draw
  (1.25,0.375) node {$(\{a,c\}\{a\})^\omega$};
\draw
  (1.75,0.375) node {$\{b\}\{a,c\}^\omega$};

\draw
  (0,0.9) node[anchor=west] {$\partition_{\emptyset, \emptyset}$}
  (0,0.1) node[anchor=west] {$\partition_{\emptyset, \{\G a\}} $}
  (2,0.9) node[anchor=east] {$\partition_{\{b \U c\},\emptyset}$}
  (2,0.1) node[anchor=east] {$\partition_{\{b \U c\},\{\G a\}} $}; 
\end{tikzpicture}
  \end{center}
  \caption{Partition of $(2^{\{a,b,c\}})^\omega$ according to $\varphi = \G a \vee b \U c$.}
  \label{fig:ex:partition}
\end{figure}

\paragraph{Stable words.} A word $w$ is \emph{stable} with respect to $\varphi$ if
every formula of $\sfmu(\varphi)$ holds either never or infinitely often along $w$ (i.e., either 
none or infinitely many of its suffixes satisfy the formula), and every formula of $\sfnu(\varphi)$
fails never or infinitely often along $w$. In particular, for a stable word no formula of $\sfmu(\varphi)$ can hold a finite, nonzero number of times before it fails forever, and no formula of $\sfnu(\varphi)$ can fail a finite, nonzero number of times before it holds forever.  It follows immediately from this definition that not every word is stable, but every word eventually stabilises, meaning that all but finitely many of its suffixes are stable. 
Let $\stablewords_\varphi$ denote the set of stable words with respect to $\varphi$. Defining 
\begin{align*}
\setF^{\;\!\varphi}_w  & \coloneqq \{\psi : \psi \in \sfmu(\varphi) \wedge w \models \F\psi \}  \\
\setG^{\varphi}_w  & \coloneqq \{\psi : \psi \in \sfnu(\varphi) \wedge w \models \G\psi \} 
\end{align*}
\noindent we easily obtain:
\begin{equation}
\label{eq:stability}
\stablewords_\varphi  \coloneqq \{ w \in \Univ : \setF^{\;\!\varphi}_w = \setGF^{\;\varphi}_w  \; \wedge \; \setG^\varphi_w = \setFG^\varphi_w \}
\end{equation}

\begin{example}
Let $\varphi = \G a \vee b \U c$. The words $\{c\}^n \{a\}^\omega$ for $n \geq 1$ are not stable w.r.t. $\varphi$, because $b \U c$ holds exactly $n$ times along the word. However, the suffix $\{a\}^\omega$ is stable.
\Cref{fig:ex:partition} represents the stable words of each element $\partition_{\setmu,\setnu}$ of the partition in darker colour, and gives examples of stable words for each class.
\end{example}

The starting point of this paper is the observation that some results of \cite{DBLP:conf/lics/EsparzaKS18,DBLP:phd/dnb/Sickert19} allow us to easily derive a normal form
for LTL, albeit only when LTL is interpreted on stable words. More precisely, in \Cref{sec:eval} we show that for every 
$\setmu \subseteq \sfmu(\varphi)$ and $\setnu \subseteq \sfnu(\varphi)$ there exist formulas $\evalnu{\varphi}{\setmu} \in \Pi_1$ and $\evalmu{\varphi}{\setnu} \in \Sigma_1$ such that:
{\small
\begin{equation}
\label{eq:stableeq}
\varphi \! \equ{\stablewords_\varphi} \!\!\! \bigvee_{\substack{\setmu \subseteq \sfmu(\varphi)\\\setnu \subseteq \sfnu(\varphi)}} \left( \evalnu{\varphi}{\setmu} \! \wedge \! \bigwedge_{\psi \in \setmu} \!\! \G\F(\evalmu{\psi}{\setnu}) \wedge \bigwedge_{\psi \in \setnu} \!\! \F\G(\evalnu{\psi}{\setmu}) \right) \\
\end{equation}
}

\noindent Further, $\evalnu{\varphi}{\setmu}$ and $\evalmu{\varphi}{\setnu}$ are obtained from $\varphi$, $\setmu$, and $\setnu$ by means of a simple, linear-time syntactic substitution procedure. Observe that the right-hand side is a formula of $\Delta_2$, and that we write $\equ{\stablewords_\varphi}$, i.e., the equivalence is only valid within the universe of stable words. In this paper we lift this restriction. In \Cref{sec:newnf} we define a formula $\flatten{\varphi}{\setmu} \in \Sigma_2$ by means of another linear-time, syntactic substitution procedure, such that:
{\small
\begin{equation}
\label{eq:generaleq}
\varphi  \equiv  \bigvee_{\substack{\setmu \subseteq \sfmu(\varphi)\\\setnu \subseteq \sfnu(\varphi)}} \left( \flatten{\varphi}{\setmu} \wedge \bigwedge_{\psi \in \setmu} \G\F(\evalmu{\psi}{\setnu}) \wedge \bigwedge_{\psi \in \setnu} \F\G(\evalnu{\psi}{\setmu}) \right)
\end{equation}
}

\begin{example}
For $\varphi = \F( a \wedge \G( b \vee \F c)) \in \Sigma_3$, the still-to-be-defined normal form (\ref{eq:stableeq}) will yield:
\[\varphi \equ{\stablewords_\varphi} (\G\F a \wedge \F\G b) \vee (\G\F a \wedge \G\F c)\] 
\noindent Indeed, since $\varphi \in \sfmu(\varphi)$, every stable word satisfying $\varphi$ must satisfy it infinitely often, and so equivalence for stable words holds, although the formulas are not equivalent. For Equation (\ref{eq:generaleq}) we will obtain:
\[\varphi \equiv \F (a \wedge ((b \vee \F c) \; \U \; \G b))  \vee (\F a \wedge \G\F c)\]
Observe that the right-hand-side belongs to $\Delta_2$.
\end{example}

\section{The Formulas $\evalnu{\varphi}{\setmu}$ and $\evalmu{\varphi}{\setnu}$}
\label{sec:eval}

We recall the definitions of the formulas $\evalnu{\varphi}{\setmu}$ and $\evalmu{\varphi}{\setnu}$, introduced in \cite{DBLP:conf/lics/EsparzaKS18,DBLP:phd/dnb/Sickert19} with a slightly different notation. 

\paragraph{The formula $\evalnu{\varphi}{\setmu}$.}  Define $\partition_\setmu \coloneqq \bigcup_{\setnu \subseteq \sfnu(\varphi)} \partition_{\setmu,\setnu}$. Observe that $\partition_\setmu$ is the language of the words $w$ such that $\setmu = \setGF_w$. The formula 
$\evalnu{\varphi}{\setmu}$ is defined with the goal of satisfying the following identity:
 \begin{equation}
 \label{eq:specmu}
 \varphi \equ{\stablewords_\varphi \cap \partition_\setmu} \evalnu{\varphi}{\setmu}
 \end{equation}
\noindent Intuitively, the identity states that within the universe of the stable words of $\partition_\setmu$, the formula $\varphi$ can be replaced by the simpler formula $\evalnu{\varphi}{\setmu}$. 

All insights required to define  $\evalnu{\varphi}{\setmu}$ are illustrated by the following examples, where we assume that $w \in \stablewords_\varphi \cap \partition_\setmu$:
\begin{itemize}
\item $\varphi = \F a \wedge \G b$ and $\setmu = \{ \F a \}$. Since $\setmu = \setGF_w$, we have  $\F a \in \setGF_w$, which implies $w \models \G\F a$. So $w \models \F a \wedge \G b$ if{}f $w \models \G b$, and so we can set $\evalnu{\varphi}{\setmu} \coloneqq \true \wedge \G b$, i.e., we can define $\evalnu{\varphi}{\setmu}$ as the result of substituting $\true$ for $\F a$ in $\varphi$. The yet-to-be-defined substitution in-fact replaces the abbreviation $\F a = \true \U a$ by $\true \W a \equiv \true$.
\item $\varphi = \F a \wedge \G b$ and $\setmu = \emptyset$. Since $\setmu = \setF_w$, we have $\F a \notin \setF_w$, and so $w \not\models \F a$. In other words, $w \models \F a \wedge \G b$ if{}f $w \models \false$, and so we can set $\evalnu{\varphi}{\setmu} \coloneqq \false \wedge \G b$.
\item $\varphi = \G(b \U c)$ and $\setmu=\{ b \U c\}$. Since $\setmu = \setGF_w$, we have  $b \U c \in \setGF_w$, and so $w \models \G\F (b \U c)$. This does not imply $w_i \models b \U c$ for all suffixes of $w$, but it implies that $c$ will hold infinitely often in the future. So $w \models \G(b \U c)$ if{}f $w \models \G(b \W c)$, and so we can define $\evalnu{\varphi}{\setmu} \coloneqq \G(b \W c)$.
\end{itemize}

\begin{definition}[\cite{DBLP:conf/lics/EsparzaKS18,DBLP:phd/dnb/Sickert19}]\label{def:evalnu}
Let $\setmu \subseteq \sfmu(\varphi)$ be a set of formulas. The formula $\evalnu{\varphi}{\setmu}$ is inductively defined as follows:
\begin{align*}
\evalnu{(\varphi \U \psi)}{\setmu} & \coloneqq \begin{cases} \evalnu{\varphi}{\setmu} \; \W \; \evalnu{\psi}{\setmu}\hphantom{\R} & \mbox{if $\varphi \U \psi \in \setmu$} \\ \false & \mbox{otherwise.}\end{cases} \\
\evalnu{(\varphi \M \psi)}{\setmu} & \coloneqq \begin{cases} \evalnu{\varphi}{\setmu} \; \R \; \evalnu{\psi}{\setmu} \hphantom{\W} & \mbox{if $\varphi \M \psi \in \setmu$} \\ \false & \mbox{otherwise.}\end{cases} 
\end{align*} %
All other cases are defined homomorphically, e.g., $\evalnu{a}{\setmu} \coloneqq$ $a$ for every $a \in Ap$, 
$\evalnu{(\X \varphi)}{\setmu} \coloneqq \X(\evalnu{\varphi}{\setmu})$, and $\evalnu{(\varphi \W \psi)}{\setmu}$ $\coloneqq $ $(\evalnu{\varphi}{\setmu}) \, \W \, (\evalnu{\psi}{\setmu})$.
\end{definition}

The following lemma, proved in \cite{DBLP:conf/lics/EsparzaKS18,DBLP:phd/dnb/Sickert19}, shows that $\evalnu{\varphi}{\setmu}$ indeed satisfies \Cref{eq:specmu}. Since the notation of \cite{DBLP:conf/lics/EsparzaKS18,DBLP:phd/dnb/Sickert19} is slightly different, we include proofs with the new notation for the cited results in 
\ifarxiv
\Cref{sec:updated_proofs} for convenience.
\fi 

\begin{restatable}[\cite{DBLP:conf/lics/EsparzaKS18,DBLP:phd/dnb/Sickert19}]{lemma}{lemEvalnu}\label{lem:evalnu}
Let $w$ be a word, and let $\setmu \subseteq \sfmu(\varphi)$ be a set of formulas.
	\begin{enumerate}
		\item If $\setF_w^{\;\!\varphi} \subseteq \setmu$ and $w \models \varphi$, then $w \models \evalnu{\varphi}{\setmu}$.
		\item If $\setmu \subseteq \setGF_w^{\;\varphi}$ and $w \models \evalnu{\varphi}{\setmu}$, then $w \models \varphi$.
	    \item $\varphi \equ{\stablewords_\varphi \cap \partition_\setmu} \evalnu{\varphi}{\setmu}$
	\end{enumerate}
\end{restatable}

Observe that the first two statements do not assume that $w$ is stable. This is an aspect we will later make use of for the definition of the normalisation procedure.

\paragraph{The formula $\evalmu{\varphi}{\setnu}$.} Let $\partition_\setnu \coloneqq \bigcup_{\setmu \subseteq \sfmu(\varphi)} \partition_{\setmu,\setnu}$. The formula $\evalmu{\varphi}{\setnu}$ is designed to satisfy 
\begin{equation}\label{eq:specnu}
\varphi \equiv^{\stablewords_\varphi \cap \partition_\setnu} \evalmu{\varphi}{\setnu}
\end{equation}
\noindent and its definition is completely dual to that of $\evalnu{\varphi}{\setmu}$.
\begin{definition}[\cite{DBLP:conf/lics/EsparzaKS18,DBLP:phd/dnb/Sickert19}]\label{def:evalmu}
Let $\setnu \subseteq \sfnu(\varphi)$ be a set of formulas. The formula $\evalmu{\varphi}{\setnu}$ is inductively defined as follows:
\begin{align*}
\evalmu{(\varphi \R \psi)}{\setnu} & = \begin{cases} \true & \mbox{if $\varphi \R \psi \in \setnu$} \\ \evalmu{\varphi}{\setnu} \; \M \; \evalmu{\psi}{\setnu} \hphantom{\U} & \mbox{otherwise.}\end{cases} \\
\evalmu{(\varphi \W \psi)}{\setnu} & = \begin{cases} \true & \mbox{if $\varphi \W \psi \in \setnu$} \\ \evalmu{\varphi}{\setnu} \; \U \; \evalmu{\psi}{\setnu} \hphantom{\M} & \mbox{otherwise.}\end{cases}
\end{align*} %
All other cases are defined homomorphically.
\end{definition}

\noindent The dual of \Cref{lem:evalnu} also holds:

\begin{restatable}[\cite{DBLP:conf/lics/EsparzaKS18,DBLP:phd/dnb/Sickert19}]{lemma}{lemEvalmu}\label{lem:evalmu}
Let $w$ be a word, and let $\setnu \subseteq \sfnu(\varphi)$ be a set of formulas.
	\begin{enumerate}[resume]
		\item If $\setFG_w^{\,\varphi} \subseteq \setnu$ and $w \models \varphi$, then $w \models \evalmu{\varphi}{\setnu}$.
		\item If $\setnu \subseteq \setG_w^{\varphi}$ and  $w \models \evalmu{\varphi}{\setnu}$, then $w \models \varphi$.
		\item $\varphi \equ{\stablewords_\varphi \cap \partition_\setnu} \evalmu{\varphi}{\setnu}$
	\end{enumerate}
\end{restatable}

\paragraph{A normal form for stable words.} We use the following result from \cite{DBLP:conf/lics/EsparzaKS18,DBLP:phd/dnb/Sickert19} to characterise the stable words of a partition $\partition_{M,N}$ that satisfy $\varphi$:

\begin{restatable}[\cite{DBLP:conf/lics/EsparzaKS18,DBLP:phd/dnb/Sickert19}]{lemma}{lemInduction}\label{lem:evalFG:setFG1}
Let $w$ be a word, and let $\setmu \subseteq \sfmu(\varphi)$ and $\setnu \subseteq \sfnu(\varphi)$. Then define:
\[\Phi(\setmu, \setnu) \coloneqq \bigwedge_{\psi \in \setmu} \G\F(\evalmu{\psi}{\setnu}) \wedge \bigwedge_{\psi \in \setnu} \F\G(\evalnu{\psi}{\setmu})\]
\noindent We have:
\begin{enumerate}
\item If $\setmu = \setGF_w$ and $\setnu = \setFG_w$, then $w \models \Phi(\setmu, \setnu)$.
\item If $w \models \Phi(\setmu, \setnu)$, then $\setmu \subseteq \setGF_w$ and $\setnu \subseteq \setFG_w$.
\end{enumerate}
\end{restatable}

Equipped with this lemma, let us show that a stable word of $\partition_{\setmu,\setnu}$  satisfies $\varphi$ if{}f it satisfies $\evalnu{\varphi}{\setmu} \wedge \Phi(\setmu,\setnu)$. Let $w$ be a stable word of $\partition_{\setmu,\setnu}$. If $w$ satisfies $\varphi$, then it satisfies $\evalnu{\varphi}{\setmu}$ by \Cref{lem:evalnu}.3 and $\Phi(\setmu,\setnu)$ by \Cref{lem:evalFG:setFG1}.1 (recall that, since $w \in \partition_{\setmu,\setnu}$, we have
$\setmu =\setGF_w$ and $\setnu =\setGF_w$ by Equation (\ref{eq:partition})). For the other direction, assume that
$w$ satisfies $\evalnu{\varphi}{M} \wedge \Phi(\setmu,\setnu)$. Then we have 
$\setmu \subseteq \setGF_w$ by \Cref{lem:evalFG:setFG1}.2 and so $w$ satisfies $\varphi$
by \Cref{lem:evalnu}.2. (This direction does not even require stability.)

Since every word belongs to some element of the partition, we obtain a normal form for stable words:

\begin{proposition}
\label{prop:nf1}
\[
\varphi \! \equ{\stablewords_\varphi} \!\!\! \bigvee_{\substack{\setmu \subseteq \sfmu(\varphi)\\\setnu \subseteq \sfnu(\varphi)}} \left( \evalnu{\varphi}{\setmu} \! \wedge \! \bigwedge_{\psi \in \setmu} \!\! \G\F(\evalmu{\psi}{\setnu}) \wedge \bigwedge_{\psi \in \setnu} \!\! \F\G(\evalnu{\psi}{\setmu}) \right) 
\]
\end{proposition}
\begin{proof}
Define $\Phi(\setmu, \setnu)$ as in \Cref{lem:evalFG:setFG1} and let $w \in \stablewords_\varphi$ be a stable word. We show that $w$ satisfies $\varphi$ if{}f it satisfies $\evalnu{\varphi}{\setmu}$ and $\Phi(\setmu, \setnu)$ for some $\setmu \subseteq \sfmu(\varphi)$ and $\setnu \subseteq \sfnu(\varphi)$.

Assume $w \models \varphi$. Let $M\coloneqq\setGF_w$ and $N\coloneqq\setFG_w$.
By \Cref{lem:evalFG:setFG1}.1 $w \models \Phi(M,N)$ holds. 
Since $w$ is stable, we have $\setF_w = \setGF_w = M$ (see \Cref{eq:stability}). 
By \Cref{lem:evalnu}.1 we have $w \models \evalnu{\varphi}{M}$, and we are done.

Assume $w \models \left(\evalnu{\varphi}{\setmu} \wedge \Phi(\setmu, \setnu)\right)$ for some $\setmu \subseteq \sfmu(\varphi)$ and $\setnu \subseteq \sfnu(\varphi)$. Using the second part of \Cref{lem:evalFG:setFG1} we get $\setmu \subseteq \setGF_w$. Applying \Cref{lem:evalnu}.2 we get $w \models \varphi$.
\end{proof}

\begin{example}
\label{ex:oldnf}
Let $\varphi = \F( a \wedge \G( b \vee \F c))$. 
We have $\sfmu(\varphi) = \{ \varphi, \F c\}$ and $\sfnu(\varphi) = \{\G( b \vee \F c)\}$. So there are four possible choices for $\setmu$, and two for $\setnu$. It follows that the right-hand-side of \Cref{prop:nf1} has eight disjuncts. However, all disjuncts with $\varphi \notin \setmu$ are equivalent to $\false$ because then $\evalnu{\varphi}{\setmu} = \false$, and the same holds for all disjuncts with $\varphi \in \setmu$ and $\setnu = \emptyset$ because $\evalmu{\varphi}{\emptyset} = \false$.

The two remaining disjuncts are $\setmu_1 = \{\varphi\}$, $\setnu_1 = \{\G( b \vee \F c)\}$, and 
$\setmu_2 = \{\varphi, \F c\}$, $\setnu_2 =\{\G( b \vee \F c)\}$. For both we have $\evalnu{\varphi}{\setmu_1} \equiv \evalnu{\varphi}{\setmu_2} \equiv \true$. Further, for the first disjunct we have
\[\G\F (\evalmu{\varphi}{\setnu_1}) \wedge \F\G (\evalnu{(\G( b \vee \F c))}{\setmu_1}) \equiv \G\F a \wedge \F\G b\]
\noindent and for the second we get
\begin{align*}
          & \G\F (\evalmu{\varphi}{\setnu_2}) \wedge \G\F (\evalmu{(\F c)}{\setnu_2}) \wedge \F\G (\evalnu{(\G( b \vee \F c))}{\setmu_2})  \\
\equiv \; & \G\F a \wedge \G\F c \wedge \F\G (\G \true) \equiv \G\F a \wedge \G\F c \ .
\end{align*}
\noindent Together we obtain  $\; \F( a \wedge \G( b \vee \F c)) \equ{\stablewords_\varphi} \G\F a \wedge (\F\G b \vee \G\F c)$.
\end{example}
\renewcommand{\flat}[2]{#1\langle #2 \rangle}

\section{A Normal Form for LTL}
\label{sec:newnf}
\Cref{prop:nf1} has little interest in itself because of the restriction to stable words. 
However, it serves as the starting point for our search for an unrestricted normal form, valid for all words. 
Observe that \Cref{lem:evalFG:setFG1} does not depend on $w$ being stable. Contrary, \Cref{lem:evalnu}.1 refers to $\setF_w$ and we crucially depend on stability to replace it by $\setGF_w$.
Consequently, we only need to find a replacement for the first conjunct and can leave the rest of the structure, i.e. the enumeration of all possible combinations of $\setmu \subseteq \sfmu(\varphi)$ and  $\Phi(M, N)$, unchanged.
More precisely, we search for a mapping $\flat{\varphi}{\cdot}$ that assigns to every $\setmu \subseteq \sfmu(\varphi)$ a formula $\flat{\varphi}{\setmu} \in \Sigma_2$ such that:
{\small
\begin{equation}
\label{eq:goal}
\varphi \equiv \bigvee_{\substack{\setmu \subseteq \sfmu(\varphi)\\\setnu \subseteq \sfnu(\varphi)}} \left( \flat{\varphi}{\setmu} \wedge \bigwedge_{\psi \in \setmu} \G\F(\evalmu{\psi}{\setnu}) \wedge \bigwedge_{\psi \in \setnu} \F\G(\evalnu{\psi}{\setmu}) \right)
\end{equation}}

\noindent
The following lemma gives sufficient conditions for $\flat{\varphi}{\setmu}$.

\begin{lemma}\label{lem:ltl_equiv}
For every $M \subseteq \sfmu(\varphi)$, let $\flat{\varphi}{M}$ be a formula satisfying:
\begin{itemize}
\item[(a)] For every $M' \subseteq \sfmu(\varphi)$: 
$M \subseteq M' \implies \flat{\varphi}{M} \models \flat{\varphi}{M'}$
\item[(b)] For every word $w$: 
$w \models \varphi \iff w \models \flat{\varphi}{\setGF^{\;\varphi}_w}$
\end{itemize}
Then \Cref{eq:goal} holds. 
\end{lemma}
\begin{proof}
Assume that $(a,b)$ hold, and let $w$ be a word. 
We show that $w$ satisfies $\varphi$ if{}f it satisfies the 
right-hand-side of (\ref{eq:goal}).

\smallskip

\noindent ($\Rightarrow$) Assume $w$ satisfies $\varphi$. 
By (b) we have $w \models \flat{\varphi}{\setGF_w}$.
We claim that the disjunct of the right-hand-side of \Cref{eq:goal} with
$\setmu \coloneqq \setGF_w$ and $\setnu \coloneqq \setFG_w$ holds. 
Indeed, $w \models \flat{\varphi}{\setmu}$ trivially holds, and the rest follows from \Cref{lem:evalFG:setFG1}.1. 

\smallskip

\noindent ($\Leftarrow$) Assume $w$ satisfies the right-hand side of \Cref{eq:goal}. Then there exist $\setmu \subseteq \sfmu(\varphi)$ and $\setnu \subseteq \sfnu(\varphi)$ such that $w \models \flat{\varphi}{\setmu}$ holds, $w \models \G\F(\evalmu{\psi}{\setnu})$ holds for every $\psi \in \setmu$, and $w \models \F\G(\evalnu{\psi}{\setmu})$ holds for every $\psi \in \setnu$. \Cref{lem:evalFG:setFG1}.2 yields $\setmu \subseteq \setGF_w$, and (a) yields $\flat{\varphi}{\setGF_w}$.  Applying (b) we get $w \models \varphi$.
\end{proof}

Note that \Cref{lem:ltl_equiv} can also be dualised and we could search for a mapping $\flat{\varphi}{\cdot}$ that assigns to every $\setnu \subseteq \sfnu(\varphi)$ a formula $\flat{\varphi}{\setnu} \in \Pi_2$ such that \Cref{eq:goal} holds. 

Unfortunately we cannot simply take $\flat{\varphi}{M} \coloneqq \evalnu{\varphi}{M}$ or $\flat{\varphi}{N} \coloneqq \evalmu{\varphi}{N}$: Both choices satisfy condition (a) of Lemma \ref{lem:ltl_equiv}, as proven by \Cref{lem:propsofevalnu}\footnote{This lemma is needed again for the proof of \Cref{thm:logical:flatten}.}, but fail to satisfy condition (b) as shown by \Cref{ex:notusable}.

\begin{lemma}
\label{lem:propsofevalnu} $\evalnu{\varphi}{\cdot}$ and $\evalmu{\varphi}{\cdot}$ have the following properties: For every $M,M' \subseteq \sfmu(\varphi)$ and $N,N' \subseteq \sfnu(\varphi)$: 
\begin{align*}
M \subseteq M' & \implies \evalnu{\varphi}{M} \models \evalnu{\varphi}{M'} \\ 
N \subseteq N' & \implies \evalmu{\varphi}{N} \models \evalmu{\varphi}{N'}
\end{align*}
\end{lemma}

\begin{proof}
\noindent (a) By induction on $\varphi$. We show only two cases, since all other cases are either trivial or analogous.

\smallskip\noindent Case $\varphi = \psi_1 \U \psi_2$.  Assume $w \models \evalnu{\varphi}{\setmu}$ holds. Due to the definition of $\evalnu{\varphi}{\setmu}$ we have $\varphi \in \setmu$ and thus also $\varphi \in \setmu'$. Thus we have $w \models (\evalnu{\psi_1}{\setmu}) \W (\evalnu{\psi_2}{\setmu})$ and applying the induction hypothesis we get $w \models (\evalnu{\psi_1}{\setmu'}) \W (\evalnu{\psi_2}{\setmu'})$. Hence $w \models \evalnu{\varphi}{\setmu'}$.

\smallskip\noindent Case $\varphi = \psi_1 \W \psi_2$. Assume $w \models \evalmu{\varphi}{\setnu}$ holds. If $\varphi \in \setnu'$ then $w \models \evalmu{\varphi}{\setnu'}$ trivially holds. If $\varphi \notin \setnu'$ then also $\varphi \notin \setnu$, and we get $w \models (\evalmu{\psi_1}{\setnu}) \U (\evalmu{\psi_2}{\setnu})$. Using the induction hypothesis we get $w \models (\evalmu{\psi_1}{\setnu'}) \U (\evalmu{\psi_2}{\setnu'})$, and we are done. 
\end{proof}

\begin{example}
\label{ex:notusable}
Let us first exhibit a formula $\varphi$ and a word $w$ such that $w \models \varphi$, but $w \not \models \evalnu{\varphi}{\setGF^{\;\varphi}_w}$. For this take $\varphi = \F a$ and $w = \{a\}\{\}^\omega$. Thus $w \models \varphi$ and $\setGF_w = \emptyset$. However, $\evalnu{(\F a)}{\emptyset} = \false$ and hence $w \not \models \evalnu{(\F a)}{\setGF^{\;\varphi}_w}$.

We now move to the second case. Let us exhibit $\varphi$ and $w$ such that $w \not \models \varphi$ and $w \models \evalmu{\varphi}{\setFG^{\varphi}_{w}}$. Dually, let $\varphi = \G a$ and $w = \{\}\{a\}^\omega$. Then $w \not \models \varphi$, but $\setFG_{w} = \{\G a\}$ and $\evalmu{(\G a)}{\{\G a\}} = \true$ and hence $w \models \evalnu{(\G a)}{\setFG^{\varphi}_{w}}$.
\end{example}

The key to finding a mapping $\flat{\varphi}{\cdot}$ satisfying both conditions of
\Cref{lem:ltl_equiv} is the technical result below, for which we offer the following intuition. The following equivalence is a valid law of LTL: 
\begin{align} \G \varphi \equiv \varphi \, \U \, \G \varphi \end{align}
In order to prove that a word $w$  satisfies the right-hand-side we can take an \emph{arbitrary} index $i \geq 0$, prove that $w_j \models \varphi$ holds for every $j< i$, and then prove that $w_i \models \G \varphi$. Since we are free to choose $i$, we can pick it such that $w_i$ is a stable word, which allows us to apply the machinery of Section \ref{sec:eval} and obtain:

\begin{lemma}\label{lem:flatten:correct-local}
For every word $w$:  
\[
w \models \G \varphi  \iff  w \models \varphi \; \U \; \G (\evalnu{\varphi}{\setGF^{\;\varphi}_w}) 
\]
\end{lemma}

\begin{proof}
We prove both directions separately.

\smallskip\noindent 
($\Rightarrow$) Assume $w \models \G \varphi$ holds. Let $w_i$ be a stable suffix of $w$. By the definition of stability we have $\setF^{\;\!\varphi}_{w_i} = \setF^{\;\!\varphi}_{w_j} = \setGF^{\;\varphi}_{w}$ for every $j \geq i$. By \Cref{lem:evalnu}.1, we have 
\[w_j \models \varphi \implies w_j \models \evalnu{\varphi}{\setGF^{\;\varphi}_{w}} \mbox{ for every $j \geq i$}\]
\noindent and so in particular $w_i \models \G (\evalnu{\varphi}{\setGF^{\;\varphi}_{w}})$. We proceed as follows:
\[\def\arraystretch{1.2}
\begin{array}{rl}
               & w \models \G \varphi \\
\implies       & w_i \models \G (\evalnu{\varphi}{\setGF^{\;\varphi}_{w}}) \wedge \forall k < i. ~ w_k \models \varphi \\
\implies       & w \models \varphi \; \U \; \G(\evalnu{\varphi}{\setGF^{\;\varphi}_{w}}) \\ 
\end{array}\]


\noindent 
($\Leftarrow$) This is an immediate consequence of \Cref{lem:evalnu}.2. 
\end{proof}

With the help of the standard LTL-equivalences
\begin{align}
\varphi \W \psi & \equiv \varphi \U (\psi \vee \G \varphi) \label{law1} \\
\varphi \R \psi & \equiv (\varphi \vee \G \psi) \M \psi \label{law2}
\end{align}
\noindent \Cref{lem:flatten:correct-local} can be extended to a more powerful proposition.

\begin{proposition}\label{prop:flatten:correct-local}
For all formulas $\varphi$, $\psi$, and for every word $w$:  
\[\begin{array}{lcl}
w \models \varphi \W \psi & \iff & w \models \varphi \; \U \; \big(\psi \vee \G (\evalnu{\varphi}{\setGF^{\;\varphi}_w})\big) \\
w \models \varphi \R \psi & \iff & w \models \big(\varphi \vee \G (\evalnu{\psi}{\setGF^{\;\psi}_w})\big) \; \M \; \psi
\end{array}\]
\end{proposition}
\begin{proof}

We only prove the first statement. The proof of the second is dual.

\smallskip\noindent 
($\Rightarrow$) Assume $w \models \varphi \W \psi$. We split this branch of the proof further, by a case distinction on whether $w \models \G \varphi$ holds. If $w \models \G \varphi$ holds, then by Lemma \ref{lem:flatten:correct-local} we have
$w \models \varphi \; \U \; $ $\G(\evalnu{\varphi}{\setGF^{\;\varphi}_{w}})$, and so
$w \models \varphi \; \U \; (\psi \vee \G (\evalnu{\varphi}{\setGF^{\;\varphi}_w}))$ holds.
Assume now that $w \not \models \G \varphi$. Then we simply derive:
\[ \def\arraystretch{1.2}
\begin{array}{rlr}
           & w \models \varphi \W \psi \\
\iff       & w \models \varphi \U \psi & \text{($w \not \models \G \varphi$)} \\           
\implies   & w \models \varphi \U \; (\psi \vee \G (\evalnu{\varphi}{\setGF^{\;\varphi}_w}))
\end{array}\]

\noindent 
($\Leftarrow$)  By \Cref{lem:evalnu}.2 we have $(w_j \models \evalnu{\varphi}{\setGF^{\;\varphi}_w} \implies w_j \models \varphi)$ for all $j \geq 0$. Thus $w_j \models \evalnu{(\G \varphi)}{\setGF^{\;\varphi}_w} \implies w_j \models \G \varphi$ for all $j \geq 0$ and we can simply derive:
\[\arraycolsep=3.6pt\def\arraystretch{1.2}
\begin{array}{rlrc}
           & w \models \varphi \U (\psi \vee \G(\evalnu{\varphi}{\setGF^{\;\varphi}_w})) \\
\implies   & w \models \varphi \U (\psi \vee \G \varphi)   & \text{(\Cref{lem:evalnu}.2)}  \\
\iff       & w \models \varphi \W \psi                     & \text{(\Cref{law1})} & \qedhere
\end{array}\]
\end{proof}

Proposition \ref{prop:flatten:correct-local} gives us all we need to define a formula $\flatten{\varphi}{\setmu}$ satisfying \Cref{eq:goal}.

\begin{definition}\label{def:evalmunu}
Let $\varphi$ be a formula and let $\setmu \subseteq \sfmu(\varphi)$. The formula $\flatten{\varphi}{\setmu}$ is inductively defined as follows for $\R$ and $\W$
\begin{align*}
\flatten{(\varphi \R \psi)}{\setmu} & = (\flatten{\varphi}{\setmu} \vee \G (\evalnu{\psi}{\setmu}))\; \M \; \flatten{\psi}{\setmu} \\
\flatten{(\varphi \W \psi)}{\setmu} & = \flatten{\varphi}{\setmu} \; \U \; (\flatten{\psi}{\setmu} \vee \G (\evalnu{\varphi}{\setmu}))
\end{align*} %
\noindent and homomorphically for all other cases.
\end{definition}

A straightforward induction on $\varphi$ shows that $\flatten{\varphi}{\setmu} \in \Sigma_2$, justifying our notation. We prove that $\flatten{\varphi}{\setmu}$ satisfies (\ref{eq:goal}) by checking that it satisfies the conditions of \Cref{lem:ltl_equiv}.

\begin{theorem}\label{thm:logical:flatten}
Let $\varphi$ be a formula. Then:
\[
\varphi \equiv \bigvee_{\substack{\setmu \subseteq \sfmu(\varphi)\\\setnu \subseteq \sfnu(\varphi)}} \left( \flatten{\varphi}{\setmu} \wedge \bigwedge_{\psi \in \setmu} \G\F(\evalmu{\psi}{\setnu}) \wedge \bigwedge_{\psi \in \setnu} \F\G(\evalnu{\psi}{\setmu}) \right)
\]
\end{theorem}

\begin{proof}
We show that conditions (a) and (b) of \Cref{lem:ltl_equiv} hold. 

\smallskip\noindent (a) The proof is an easy induction on $\varphi$, applying \Cref{lem:propsofevalnu} where necessary. 

\smallskip\noindent (b) We prove that  
\begin{equation}
\label{eq:equivw2} 
\forall w. ~ w \models \varphi \iff w \models \flatten{\varphi}{\setGF^{\;\varphi}_w}
\end{equation}
\noindent holds by structural induction on $\varphi$. We make use of the  identity
\begin{align}
\flatten{\psi}{\setmu} = \flatten{\psi}{\setmu \cap \sfmu(\psi)} \label{flatten:restrict}
\end{align}
\noindent which follows immediately from the fact that formulas in $M \setminus \sfmu(\psi)$ 
are not subformulas of $\psi$.

The base of the induction is $\varphi \in \{ \true, \false, a, \neg a\}$. In all these cases we have $\varphi = \flatten{\varphi}{\setGF_w}$ by definition, and so (\ref{eq:equivw2}) holds vacuously. All other cases in which $\flatten{\varphi}{\setmu}$ is defined homomorphically are handled in the same way. We consider only one of them:

\smallskip
\noindent Case $\varphi = \psi_1 \U \psi_2$. By assumption, the induction hypothesis 
(\ref{eq:equivw2}) holds for $\psi_1$ and $\psi_2$, giving:
\begin{align}
\forall u. ~ (u \models \psi_1 \iff u \models \flatten{\psi_1}{\setGF^{\;\psi_1}_u}) \label{eq:ui} \\ 
\forall v. ~ (v \models \psi_2 \iff v \models \flatten{\psi_2}{\setGF^{\;\psi_2}_v}) \label{eq:vj}
\end{align}

In order to use these two equivalences for the induction step, we need to replace $\setGF^{\;\psi_1}_u$ and $\setGF^{\;\psi_2}_v$ by $\setGF^{\;\varphi}_w$ in the context of $\flatten{\cdot}{\cdot}$. For this we instantiate $u \coloneqq w_i$ and $v\coloneqq w_j$ for arbitrary $i, j \geq 0$ in (\ref{eq:ui}) and (\ref{eq:vj}). With this choice $u$ and $v$ are suffixes of $w$, and so thus we get $\setGF^{\;\varphi}_u = \setGF^{\;\varphi}_v = \setGF^{\;\varphi}_w$. Notice further that, by intersection with $\sfmu(\cdot)$, we have $\setGF^{\;\psi_1}_u = \setGF^{\;\varphi}_w \cap \sfmu(\psi_1)$ and $\setGF^{\;\psi_2}_u = \setGF^{\;\varphi}_w \cap \sfmu(\psi_2)$. From (\ref{flatten:restrict}) we obtain:
\begin{align}
\forall i. ~ (w_i \models \psi_1 \iff w_i \models \flatten{\psi_1}{\setGF^{\;\varphi}_w}) \label{eq:wi} \\ 
\forall j. ~ (w_j \models \psi_2 \iff w_j \models \flatten{\psi_2}{\setGF^{\;\varphi}_w}) \label{eq:wj}
\end{align}

\noindent Applying (\ref{eq:wi}) and (\ref{eq:wj}) we get:
\[\arraycolsep=4.1pt\begin{array}{ll}
& w \models \psi_1 \U \psi_2 \\
\iff & \exists k. ~ w_k \models \psi_2 \wedge (\forall \ell < k. ~ w_\ell \models \psi_1)  \\
\iff & \exists k. ~ w_k \models \flatten{\psi_2}{\setGF^{\;\varphi}_w} \wedge (\forall \ell < k. ~ w_\ell \models \flatten{\psi_1}{\setGF^{\;\varphi}_w})  \\ 
\iff & w \models \flatten{(\psi_1 \U \psi_2)}{\setGF^{\;\varphi}_w}
\end{array}\]
\noindent which concludes the proof.

The remaining cases are $\varphi = \psi_1 \R \psi_2$ and $\varphi = \psi_1 \W \psi_2$.
Again, we only consider one of them, the other one being analogous.

\smallskip\noindent Case $\varphi = \psi_1 \W \psi_2$. 
The argumentation is only slightly more complicated than that of the $\psi_1 \U \psi_2$ case. By induction hypothesis (\ref{eq:wi}) and (\ref{eq:wj}) hold. With the help of \Cref{lem:flatten:correct-local} we derive:
\[\def\arraystretch{1.2}
\begin{array}{rl}
           & w \models \psi_1 \W \psi_2 \\
\iff       & w \models \psi_1 \U (\psi_2 \vee \G (\evalnu{\psi_1}{\setGF^{\;\psi_1}_w})) \hfill \text{(\Cref{prop:flatten:correct-local})} \\
\iff       & w \models \psi_1 \U (\psi_2 \vee \G (\evalnu{\psi_1}{\setGF^{\;\varphi}_w})) \\ & \hfill \text{($\evalnu{\psi}{\setmu} = \evalnu{\psi}{\setmu \cap \sfmu(\psi)}$)} \\
\iff       & w \models \; \flatten{\psi_1}{\setGF^{\;\varphi}_w} \; \U \; (\flatten{\psi_2}{\setGF^{\;\varphi}_w} \vee \G(\evalnu{\psi_1}{\setGF^{\;\varphi}_w})) \\ & \hfill \text{((\ref{eq:wi}) and (\ref{eq:wj}))} \\
\iff   & w \models \; \flatten{(\psi_1 \W \psi_2)}{\setGF^{\;\varphi}_w}  
\end{array}\]
\end{proof}

\begin{example}
\label{ex:newnf}
Let $\varphi = \F( a \wedge \G( b \vee \F c))$. We have 
$\sfmu(\varphi) = \{ \varphi, \F c\}$ and $\sfnu(\varphi) = \{\G( b \vee \F c)\}$, and so the right-hand-side of \Cref{thm:logical:flatten} has eight disjuncts.  
However, contrary to \Cref{ex:oldnf}, we have $\flatten{\varphi}{\setmu} \neq \false$ for every $M \subseteq \{\varphi, \F c\}$. Let $\Phi(M, N)$ be the disjunct for 
given sets $M$, $N$. We consider two cases:

\smallskip \noindent Case $M\coloneqq\emptyset$, $N \coloneqq\emptyset$. In this case $\Phi(\emptyset, \emptyset)
= \flatten{\varphi}{\emptyset}$, because the conjunctions over $M$ and $N$ are vacuous. We have:
\begin{align*}
\Phi(\emptyset, \emptyset) & = \flatten{\varphi}{\emptyset} \\
& = \F \left( a \wedge \left( \flatten{\G( b \vee \F c)}{\emptyset} \right) \right) \\
& = \F \left( a \wedge \left( \flatten{(( b \vee \F c) \, \W \, \false)}{\emptyset} \right) \right) \\
& = \F \left( a \wedge \left( \flatten{( b \vee \F c)}{\emptyset} \, \U \, \left( \false \vee \G(  \evalnu{(b \vee \F c)}{\emptyset}) \right)\right) \right) \\
& = \F \left( a \wedge \left( ( b \vee \F c) \, \U \, \G b \right) \right) 
\end{align*}
\smallskip \noindent  Case $M \coloneqq\{ \F c\}$, $N \coloneqq \{\G( b \vee \F c)\} $. We get:
\begin{align*}
\flatten{\varphi}{M}
& = \F \left( a \wedge \left( \flatten{( b \vee \F c)}{M} \, \U \, \left( \false \vee \G(  \evalnu{(b \vee \F c)}{M}) \right)\right) \right) \\
& = \F \left( a \wedge \left( (b \vee \F c) \, \U \, (\false \vee \true) \right) \right) = \F a
\end{align*}
\noindent Further, we have $\F\G(\evalnu{\G( b \vee \F c)}{M}) = \F\G ( \G \true) = \true$ and $\G\F( \evalmu{(\F c)}{N}) = \G\F (\F c) = \G\F c$. 
So in this case we obtain $\Phi(\{ \F c\}, \{\G( b \vee \F c)\}) = \F a \wedge \G\F c$.

\smallskip Repeating this process for all possible sets $M,N$ and bringing the resulting formula in disjunctive normal form we finally get
\[\varphi \equiv \F \left( a \wedge \left( ( b \vee \F c) \, \U \, \G b \right) \right) \vee \left( \F a \wedge \G\F c \right)\] 
\end{example}

\subsection{Complexity of the Normalisation Procedure}

We show that the normalisation procedure has at most single exponential blowup in the length of the formula, improving on the previously known non-elementary bound. 

\begin{proposition}
Let $\varphi$ be a formula with length $n$. Then there exists an equivalent formula $\varphi_{\Delta_2}$ in $\Delta_2$ of length $2^{2n + \mathcal{O}(1)}$.	
\end{proposition}

\begin{proof}
Let $\psi$ be an arbitrary formula.
We let $|\psi|$ denote the length of formula and start by giving bounds on $\evalnu{\psi}{\setmu}$, $\evalmu{\psi}{\setnu}$, and $\flatten{\psi}{\setmu}$. For this let $M \subseteq \sfmu(\psi)$ and $N \subseteq \sfnu(\psi)$ be sets of formulas. We obtain by induction on the structure of $\psi$ that $|\evalnu{\psi}{\setmu}| \leq |\psi|$, $|\evalmu{\psi}{\setnu}| \leq |\psi|$, and $|\flatten{\psi}{\setmu}| \leq 2^{|\psi| + 1}$.

Consider now the right-hand side of \Cref{thm:logical:flatten} as the postulated $\varphi_{\Delta_2}$. Using these bounds we calculate the maximal size of a disjunct and obtain:
\[2^{n+1} + n (n + 3) + n (n + 3) + 1 = 2^{n+1} + 2n^2 + 6n + 1\]	
For sufficiently large $n$, i.e. $n > 5$, we can bound this by $2^{n+2}$. There exist at most $2^n$ disjuncts and thus the formula is at most of size $2^{2n+2}$ for $n > 5$.
\end{proof}

\subsection{A Dual Normal Form}

We obtained \Cref{thm:logical:flatten} by relying on the LTL equivalence (\ref{law1}) and (\ref{law2}) for $\W$ and $\R$. Using dual LTL-equivalences for $\U$ and $\M$, $\varphi \U \psi \equiv (\varphi \wedge \F \psi) \W \psi$ and $\varphi \M \psi \equiv \varphi \R (\psi \wedge \F \varphi)$, we can also obtain a dual normalisation procedure:

\begin{definition}\label{def:evalnumu}
Let $\varphi$ be a formula and let $\setnu \subseteq \sfnu(\varphi)$ be a set of formulas. The formula $\flattentwo{\varphi}{\setnu}$ is inductively defined as follows for $\U$ and $\M$:
\begin{align*}
\flattentwo{(\varphi \U \psi)}{\setnu} & = (\flattentwo{\varphi}{\setnu} \wedge \F (\evalmu{\psi}{\setnu})) \; \W \; \flattentwo{\psi}{\setnu} \\
\flattentwo{(\varphi \M \psi)}{\setnu} & = \flattentwo{\varphi}{\setnu} \; \R \; (\flattentwo{\psi}{\setnu} \wedge \F (\evalmu{\varphi}{\setnu}))
\end{align*} %
\noindent and homomorphically for all other cases.
\end{definition}

\begin{theorem}\label{thm:logical:flatten:dual}
Let $\varphi$ be a formula. Then:
\[
\varphi \equiv \bigvee_{\substack{\setmu \subseteq \sfmu(\varphi)\\\setnu \subseteq \sfnu(\varphi)}} \left( \flattentwo{\varphi}{\setnu} \wedge \bigwedge_{\psi \in \setmu} \G\F(\evalmu{\psi}{\setnu}) \wedge \bigwedge_{\psi \in \setnu} \F\G(\evalnu{\psi}{\setmu}) \right)
\]
\end{theorem}
\section{A Translation from LTL to Deterministic Rabin Automata (DRW)}
\label{sec:LTLtoDRW}

We apply our $\Delta_2$-normalisation procedure to derive a new translation from LTL to
DRW via weak alternating automata (AWW). While the previously existing normalisation procedures could also
be used to translate LTL into DRW, the resulting DRW could have non-elementary size in the length
of the formula, making them impractical. We show that, thanks to the single exponential blow-up
of the new procedure, the new translation has double exponential blow-up, which is asymptotically optimal.

It is well-known \cite{DBLP:conf/lics/MullerSS88,DBLP:conf/tacs/Vardi94} that
an LTL formula $\varphi$ of length $n$ can be translated into an AWW with $O(n)$ states.
We show that, if $\varphi$ is in normal form, i.e., a disjunction as in \Cref{thm:logical:flatten}, then the AWW can be chosen so that every path through the automaton switches at most once between accepting and non-accepting states. We then prove that determinising AWWs satisfying this additional property is much simpler than the general case.

The section is structured as follows: \Cref{subsec:aww} introduces basic definitions,
\Cref{subsec:LTLtoAWW2} shows how to translate an $\Delta_2$-formula into AWWs with at most one switch, and
\Cref{subsec:det} presents the determinisation procedure for this subclass of AWWs.

\subsection{Weak and Very Weak Alternating Automata}
\label{subsec:aww}

Let $X$ be a finite set. The set of positive Boolean formulas over $X$, denoted by $\mathcal{B}^+(X)$, is the closure of $X \cup \{ \true, \false \}$ under disjunction and conjunction.
A set $S \subseteq X$ is a model of $\theta \in B^+(X)$ if the truth assignment that assigns true to the elements of $S$ and false to the elements of $X \setminus S$ satisfies $\theta$.
Observe, that if $S$ is a model of $\theta$ and $S \subseteq S'$ then $S'$ is also a model.
A model $S$ is  minimal if no proper subset of $S$ is a model.
The set of minimal models is denoted $\mathcal{M}_\theta$.
Two formulas are equivalent, denoted $\theta \equiv \theta'$, if their set of minimal models is equal, i.e., $\mathcal{M}_\theta = \mathcal{M}_{\theta'}$.

\paragraph{Alternating automata.}
An alternating Büchi word automaton over an alphabet $\Sigma$ is a tuple $\mathcal{A} = \langle \Sigma, Q, \theta_0, \delta, \alpha \rangle$, where $Q$ is a finite set of states, $\theta_0 \in \mathcal{B}^+(Q)$ is an initial formula, $\delta \colon Q \times \Sigma \mapsto \mathcal{B}^+(Q)$ is the transition function, and $\alpha \subseteq Q$ is the acceptance condition. A \emph{run} of $\mathcal{A}$ on the word $w$ is a directed acyclic graph $G=(V, E)$ satisfying the following properties:
\begin{itemize}
\item $V \subseteq Q \times \mathbb{N}_0$, and $E \subseteq \bigcup_{l \geq 0} ((Q \times \{l\}) \times (Q \times \{l+1\}))$.
\item There exists a minimal model $S$ of $\theta_0$ such that $(q, 0) \in V$ if{}f $q \in S$.
\item For every $(q, l) \in V$, either $\delta(q, w[l]) \equiv \false$ or the set  $\{q' \colon ((q,l),(q',l+1))\in E\}$ is a minimal model of $\delta(q, w[l])$.
\item For every $(q,l) \in V \setminus (Q \times \{0\})$ there exists $q' \in Q$ such that $((q',l-1),(q,l)) \in E$.
\end{itemize}
\noindent Runs can be finite or infinite. A run $G$ is \emph{accepting} if
\begin{itemize}
\item[(a)] $\delta(q, w[l]) \not\equiv \false$ for every $(q, l) \in V$, and
\item[(b)] every infinite path of $G$ visits $\alpha$-nodes (that is, nodes $(q, l)$ such that $q \in \alpha$) infinitely often.
\end{itemize}
In particular, every finite run satisfying (a) is accepting. $\mathcal{A}$ accepts a word $w$ iff it has an accepting run $G$ on $w$. The language $\lang(\mathcal{A})$ recognised by $\mathcal{A}$ is the set of words accepted by $\mathcal{A}$. Two automata are equivalent if they recognise the same language.

Alternating \emph{co-Büchi} automata are defined analogously, changing condition (b) by the co-Büchi condition (every infinite path of $G$ only visits $\alpha$-nodes finitely often).
Finally, in alternating \emph{Rabin} automata $\alpha$ is a set of \emph{Rabin pairs} $(F, I) \subseteq Q \times Q$, and (b) is replaced by the Rabin condition (there exists a Rabin pair $(F, I) \in \alpha$ such that every infinite path visits states of $F$ only finitely often and states of $I$ infinitely often).

An automaton is \emph{deterministic} if for every state $q \in Q$ and every letter $a \in \Sigma$ there exists $q' \in Q$ such that $\delta(q,a) = q'$, and \emph{non-deterministic} if for every $q \in Q$ and every $a \in \Sigma$ there exists $Q' \subseteq Q$ such that $\delta(q,a) = \bigvee_{q' \in Q'} q'$.

The following definitions are useful for reasoning about runs: A set $U \subseteq Q$ is called a \emph{level}. If $U \subseteq \alpha$, then $U$ is an \emph{$\alpha$-level}. A level $U'$ is a \emph{successor} of $U$ w.r.t. $a \in \Sigma$, also called $a$-successor, if for every $q \in U$ there is a minimal model $S_q$ of $\delta(q, a)$ such that $U' = \bigcup_{q \in U} S_q$. The $k$-th level of a run $G=(V, E)$ is the set $\{q \colon (q, k) \in V \}$. Observe that a level can be empty, and empty levels are $\alpha$-levels. Further, by definition a level has no successors w.r.t. $a$ if{}f it contains a state $q$ such that $\delta(q, a) \equiv \false$. In particular, every level of an accepting run has at least one successor.

\paragraph{Weak and very weak automata.} Let $\mathcal{A} = \langle \Sigma, Q, \theta_0, \delta, \alpha \rangle$ be an alternating (co-)Büchi automaton. We write $q \trans{} q'$ if there is $a \in \Sigma$ such that $q'$ belongs to some minimal model of $\delta(q, a)$. $\mathcal{A}$
is \emph{weak} if there is a partition $Q_1$, \dots, $Q_m$ of $Q$ such that
\begin{itemize}
\item for every $q, q' \in Q$, if $q \trans{} q'$  then there are $i \leq j$ such that $q \in Q_i$ and $q' \in Q_j$, and
\item for every $1 \leq i \leq m$: $Q_i \subseteq \alpha$ or $Q_i \cap \alpha = \emptyset$.
\end{itemize}
$\mathcal{A}$ is \emph{very weak} or \emph{linear} if it is weak and every class $Q_i$ of the partition is a singleton ($|Q_i| = 1$).
We let $\aww$ and $\alw$ denote the set of weak and very weak alternating automata, respectively.
Observe that for every weak automaton with a co-Büchi acceptance condition we can define a Büchi acceptance condition on the same structure recognising the same language.
Thus we will from now on assume that every weak automaton is equipped with a Büchi acceptance condition.

We define the height of a weak alternating automaton. The definition is very similar, but not identical, to the one of \cite{DBLP:conf/charme/GurumurthyKSV03}.
A weak automaton $\mathcal{A}$ has \emph{height} $n$ if every path $q \rightarrow q' \rightarrow q'' \cdots $ of $\mathcal{A}$ alternates at most $n - 1$ times between $\alpha$ and $Q \setminus \alpha$. For example, the automaton in \Cref{fig:ex:aww:1} has height 3. We let $\aww[n]$ ($\alw[n]$) denote the sets of all (very-)weak alternating automata with height at most $n$. Further, we let $\aww[n,\A]$ (resp. $\aww[n,\R]$) denote the set of automata of $\aww[n]$ whose initial formula satisfies $\theta_0 \in \mathcal{B}(\alpha)^+$ (resp. $\theta_0 \in \mathcal{B}(Q \setminus \alpha)^+$). For example the automaton of \Cref{fig:ex:aww:1} belongs to $\alw[3,\R]$.

\newcommand{\SetToWidestCases}[1]{\mathmakebox[1.1cm][l]{#1}}%

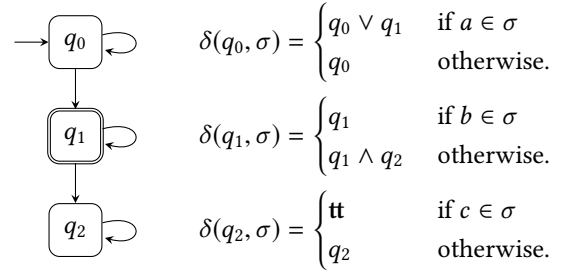
\begin{figure}
  \begin{center}
	\begin{tikzpicture}

	\node[state,initial]   (0) at (0,2.5) {$q_0$};
	\node[state,accepting] (1) at (0,1.25) {$q_1$};
	\node[state]           (2) at (0,0) {$q_2$};

	\node[] (delta0) at (4,2.5)   {$\delta(q_0, \sigma) = \begin{cases} \SetToWidestCases{q_0 \vee q_1} & \text{if } a \in \sigma \\ q_0            & \text{otherwise.} \end{cases}  $};
    \node[] (delta1) at (4,1.25) {$\delta(q_1, \sigma) = \begin{cases} \SetToWidestCases{q_1}          & \text{if } b \in \sigma \\ q_1 \wedge q_2 & \text{otherwise.} \end{cases}  $};
    \node[] (delta2) at (4,0)   {$\delta(q_2, \sigma) = \begin{cases} \SetToWidestCases{\true}        & \text{if } c \in \sigma \\ q_2            & \text{otherwise.} \end{cases}  $};

	\path[->]
	(0) edge[loop right] node[right]{} (0)
	(0) edge node[above]{} (1)
	(1) edge[loop right] node[right]{} (1)
	(1) edge node[above]{} (2)
	(2) edge[loop right] node[right]{} (2);
	\end{tikzpicture}
  \end{center}
  \caption{\alw{} for $\varphi = \F (a \wedge \X\G (b \vee \X\F c))$ with $\Sigma = 2^{\{a,b,c\}}$, $\theta_0 = q_0$, and $\alpha = \{q_1\}$.}
  \label{fig:ex:aww:1}
\end{figure}
\subsection{Translation of LTL to $\alw[2]$}
\label{subsec:LTLtoAWW2}
\newcommand{\Pairs}[1]{\mathit{An}({#1})}
\newcommand{\bsta}[1]{[{#1}]_{\leq \Gamma}}
\newcommand{\bstatwo}[2]{[{#1}]_{\leq {#2}}}
\newcommand{\sta}[2]{\langle{#1}\rangle_{#2}}

In the standard translation \cite{DBLP:conf/tacs/Vardi94} of LTL to $\alw$, the states of the $\alw$ for a formula $\varphi$ are subformulas of $\varphi$, or negations thereof. 
We show that, at the price of a slightly more complicated translation, the resulting $\alw$ for a $\Delta_i$-formula belongs to $\alw[i]$. Thus by using \Cref{thm:logical:flatten} every LTL formula can be translated to an $\alw[2]$.
The idea of the construction is to use subformulas as states ensuring that 
\begin{enumerate}
	\item the transition relation can only lead from a formula to another formula at the same level or a lower level in the syntactic-future hierarchy (\Cref{fig:syntactic-future}), and
	\item accepting states are $\Pi_i$ subformulas.
\end{enumerate}
This immediately leads to \enquote{at most one alternation}. 
However, there is a little technical problem: the level of a formula is not always well-defined, because some formulas do not belong to one single lowest level of the hierarchy. 
For example, $\X a$ belongs to both $\Pi_1$ and $\Sigma_1$. 
So we need a mechanism to \emph{disambiguate} these states. Formally we proceed as follows:

A formula is \emph{proper} if it is neither a Boolean constant ($\true$, $\false$) 
nor a conjunction or disjunction.  A \emph{state} in our modified translation is an expression of the form $\sta{\psi}{\Gamma}$, where $\psi$ is a proper formula, and $\Gamma$ 
is a smallest class of the syntactic-future hierarchy without the zeroth-level (\Cref{def:future_hierarchy}) that contains $\psi$. 
Hence we start with the classes $\Sigma_1$ and $\Pi_1$ and $\Gamma$ lies strictly above $\Delta_0$. 
Observe that for some formulas there is more than one smallest class. 
For example, since $\X a \in \Sigma_1 \cap \Pi_1$, both $\Sigma_1$ and $\Pi_1$ are smallest classes containing $\X a$, and so  both $\sta{\X a}{\Sigma_1}$ and $\sta{\X a}{\Pi_1}$ are states. For other formulas the class is unique. For example, the only state for
$a \W b$ is $\sta{ a \W b}{\Pi_1}$. 

We assign to every formula $\psi$ of LTL and every class $\Gamma$ a Boolean combination of states, denoted $\bsta{\psi}$, as follows:
\begin{itemize}
\item $\bsta{\true} = \true$ and $\bsta{\false} = \false$.
\item $\bsta{\psi_1 \vee \psi_2} = \bsta{\psi_1} \vee \bsta{\psi_2}$
\item $\bsta{\psi_1 \wedge \psi_2} = \bsta{\psi_1} \wedge \bsta{\psi_2}$
\item If $\psi$ is a proper formula, then $\bsta{\psi} = \bigvee_{\Gamma' \leq \Gamma} \, \sta{\psi}{\Gamma'}$, where $\Gamma' \leq \Gamma$ means that $\Gamma' = \Gamma$ or $\Gamma'$ is below $\Gamma$. 
\end{itemize}
For example, we obtain $\bstatwo{\X a}{\Sigma_2} = \sta{\X a}{\Sigma_1} \vee \, \sta{\X a}{\Pi_1}$ and
$\bstatwo{\X a}{\Sigma_1} = \sta{\X a}{\Sigma_1}$. 
Moreover, $\bstatwo{\F a}{\Pi_1} = \false$, since there is no $\Gamma' \leq \Pi_1$ such that $\F a \in \Gamma'$.

Let $\varphi \in \Delta_i$ for some $i \geq 0$, and let $\subf(\varphi)$ be the set of proper subformulas of $\varphi$. The automaton $\mathcal{A}_\varphi = \langle 2^{Ap}, Q, \theta_0, \delta, \alpha \rangle$ is defined as follows:
\begin{itemize}
\item $Q = \{\sta{\psi}{\Gamma} \colon \psi \in \subf(\varphi), \Gamma \leq \Delta_i\}$.
\item $\theta_0 = \bstatwo{\varphi}{\Delta_i}$.
\item $\alpha = \{ \sta{\psi}{\Pi_i} \in Q \colon i > 0 \}$.
\item $\delta$ is the restriction to $Q \times \Sigma$ of the function 
$\delta \colon \mathcal{B}^+(Q) \times \Sigma \to \mathcal{B}^+(Q)$ (notice that we overload $\delta$) defined inductively as follows: \vspace{-0.5em}
\begin{align*}
	\delta(\sta{a}{\Gamma}, \sigma) = & \begin{cases} \true & \text{if } a \in \sigma \\ \false & \text{otherwise} \end{cases} \\
	\delta(\sta{\neg a}{\Gamma}, \sigma) = & \begin{cases} \true & \text{if } a \notin \sigma \\ \false & \text{otherwise} \end{cases} \\
	\delta(\sta{\X \psi}{\Gamma}, \sigma) = & \, \bsta{\psi} \\
	\delta(\sta{\varphi \U \psi}{\Gamma}, \sigma)  = & \, \delta(\bsta{\psi \vee (\varphi \wedge \X (\varphi \U \psi))}, \sigma) \\
	\delta(\sta{\varphi \W \psi}{\Gamma}, \sigma)  = & \, \delta(\bsta{\psi \vee (\varphi \wedge \X (\varphi \W \psi))}, \sigma) \\
	\delta(\sta{\varphi \R \psi}{\Gamma}, \sigma)  = & \, \delta(\bsta{\psi \wedge (\varphi \vee \X (\varphi \R \psi))}, \sigma) \\
	\delta(\sta{\varphi \M \psi}{\Gamma}, \sigma)  = & \, \delta(\bsta{\psi \wedge (\varphi \vee \X (\varphi \M \psi))}, \sigma) 
\end{align*}
All other cases ($\true$, $\false$, $\wedge$, and $\vee$) are defined homomorphically. Observe that the $\Gamma$-bound for the $\U$, $\W$, $\R$, and $\M$ cases suffice, since every $\Gamma$ is closed under conjunction, disjunction and application of $\X$.
\end{itemize}

An example of this construction is displayed in \Cref{fig:ex:aww:1}. The states are labelled $q_0 = \sta{\varphi}{\Sigma_3}$, $q_1 = \sta{\G (b \vee \X \F c)}{\Pi_2}$, and $q_2 = \sta{\F c}{\Sigma_1}$. 

\begin{lemma}
\label{lem:aww:translation}
Let $\varphi$ be a formula of $\Delta_i$. The automaton $\mathcal{A}_\varphi$ belongs to $\alw[i]$, has $2|\subf(\varphi)|$ states, and recognises $\lang(\varphi)$.
\end{lemma}

\begin{proof}
Let us first show that $\mathcal{A}_\varphi$ belongs to $\alw[i]$. It follows immediately from the definition of $\mathcal{A}_\varphi$  that 
for every two states $\sta{\psi}{\Gamma}, \sta{\psi'}{\Gamma'}$ of $\mathcal{A}_\varphi$, if 
$\sta{\psi}{\Gamma} \trans{} \sta{\psi'}{\Gamma'}$ then $\Gamma' \leq \Gamma$. So in every path
there are at most $(i-1)$ alternations between $\Sigma$ and $\Pi$ classes.  Since the states of $\alpha$ are those 
annotated with $\Pi$ classes, there are also at most $(i-1)$ alternations between $\alpha$ and non-$\alpha$ states in a path.

To show that $\mathcal{A}_\varphi$ has at most $2|\subf(\varphi)|$ states, observe that for every formula $\psi$
there are at most two smallest classes of the syntactic-future hierarchy containing $\psi$. So $\mathcal{A}_\varphi$ has at most
two states for each formula of $\subf(\varphi)$.

To prove that $\mathcal{A}_\varphi$ recognises $\lang(\varphi)$ one shows by induction on $\psi$ that 
$\mathcal{A}_\varphi$ recognises $\lang(\psi)$ from every Boolean combination of states $\bsta{\psi}$ such that
$\psi \in \Gamma$. The proof is completely analogous to the one appearing in \cite{DBLP:conf/tacs/Vardi94}.
\end{proof}
\subsection{Determinisation of $\aww[2]$}
\label{subsec:det}

\newcommand{\Promising}{\mathit{Promising}}
\newcommand{\Levels}{\mathit{Levels}}

We present a determinisation procedure for $\aww[2,\R]$ and $\aww[2,\A]$ inspired by the break-point construction from \cite{DBLP:journals/tcs/MiyanoH84}. We only describe the construction for $\aww[2,\R]$, as the one for $\aww[2,\A]$ is dual. The following lemma states the key property of $\aww[2,\R]$:

\begin{lemma}\label{lem:rundag:normalform}
Let $\mathcal{A}$ be an $\aww[2,\R]$. $\mathcal{A}$ accepts a word $w$ if and only if there exists a run $G = (V, E)$ of $\mathcal{A}$ on $w$ such that
\begin{itemize}
\item $\delta(q, w[l]) \not \equiv \false$ for every $(q, l) \in V$, and
\item there is a threshold $k \geq 0$ such that for every $l \geq k$ and for every node $(q, l) \in V$ the state $q$ is accepting.
\end{itemize}
\end{lemma}
\begin{proof}
Assume that $\mathcal{A}$ accepts $w$. Let $G=(V, E)$ be an accepting run of $\mathcal{A}$ on $w$. Since $\mathcal{A}$ is an $\aww[2,\R]$, every path has by definition at most \emph{one} alternation of accepting and rejecting states and all states occurring in the initial formula are marked as rejecting. Hence if a node $(q, l) \in V$ is accepting, i.e. $q \in \alpha$), then all its descendants are accepting. Let $V_r \subseteq V$ be the set of rejecting nodes of $V$, i.e., the nodes $(q, l) \in V$ such that $q \notin \alpha$. Since the descendants of accepting nodes are accepting, the subgraph $G_r = (V_r, E \cap (V_r \times V_r))$ is acyclic and connected. If $V_r$ is infinite, then by Königs lemma $G_r$ has an infinite path of non-accepting nodes, contradicting that $G$ is an accepting run. So $G_r$ is finite, and we can choose the threshold $k$ as the largest level of a node of $V_r$, plus one.

Assume such a run $G=(V, E)$ exists. Condition (a) of an accepting run holds by hypothesis. For condition (b), just observe that, since the descendants of accepting nodes are accepting, and every infinite path of $G$ contains a node of the form $(q, k)$, where $k$ is the threshold level, every infinite path visits accepting nodes infinitely often.
\end{proof}

\noindent However, \Cref{lem:rundag:normalform} does not hold for $\aww[3,\R]$:

\begin{example}
Let $\mathcal{A}$ be the automaton shown in \Cref{fig:ex:aww:1} and let $w = \{a\}(\{b\}\{c\})^\omega$. Observe that $\mathcal{A}$ accepts $w$. We prove by contradiction that no run of $\mathcal{A}$ on $w$ satisfies the properties described in \Cref{lem:rundag:normalform}. Assume such a run exists. By the definition of $\delta$, the run must be infinite. Further, by assumption there exists a threshold $k$ such that all successor levels of the run are exactly $\{q_1\}$. But there exists $k' > k$ such that $w[k'] = \{c\}$. Since $\delta(q_1, \{c\}) = q_1 \wedge q_2$, the $(k' + 1)$-th level of the run contains $q_2$. Contradiction.
\end{example}

Given an automaton $\mathcal{A}$ from $\aww[2,\R]$, we construct a deterministic co-Büchi automaton $\mathcal{D}$ such that $L(\mathcal{A}) = L(\mathcal{D})$.
A state of the DCW $\mathcal{D}$ is a pair $(\Levels, \Promising)$, where $\Levels \subseteq 2^Q$ and $\Promising \subseteq 2^\alpha \cap \Levels$. 
It follows that $\mathcal{D}$ has at most $3^{2^n}$ states.
Intuitively, after reading a finite word $w_{0k} = a_0 \ldots a_k$ the automaton $\mathcal{D}$ is in the state $(\Levels_k$, $\Promising_k)$, where $\Levels_k$ contains the $k$-th levels of every run of $\mathcal{A}$ on all words with $w_{0k}$ as prefix, and $\Promising_k \subseteq \Levels_k$ contains the $\alpha$-levels of $\Levels_k$ that can still \enquote{generate} an accepting run. For this, when $\mathcal{D}$ reads $a_{i+1}$, it moves from $(\Levels_i$, $\Promising_i)$ to $(\Levels_{i+1}$, $\Promising_{i+1})$, where $\Levels_{i+1}$ contains the successors w.r.t. $a_{i+1}$ of $\Levels_i$, and $\Promising_{i+1}$ is defined as follows:
\begin{itemize}
\item If $\Promising_i \neq \emptyset$, then $\Promising_{i+1}$ contains the successors  w.r.t $a_{i+1}$ of $\Promising_i$.
\item If $\Promising_i = \emptyset$, then $\Promising_{i+1}$ contains the $\alpha$-levels of $\Levels_{i+1}$.
\end{itemize}
\noindent Finally, the co-Büchi condition contains the states $(\Levels$, $\Promising)$ such that $\Promising=\emptyset$.

Intuitively, during its run on a word $w$, the automaton $\mathcal{D}$ tracks the promising levels, removing those without successors, because they can no longer produce an accepting run. If the  $\Promising$ set becomes empty infinitely often, then every run of $\mathcal{A}$  on $w$ contains a level without successors, and so $\mathcal{A}$ does not accept $w$. If after some number of steps, say $k$, the $\Promising$ set never becomes empty again, then $\mathcal{A}$ has a run on $w$ such that every level is an $\alpha$-level and has at least one successor, and so this run is accepting.

For the formal definition of $\mathcal{D}$ it is convenient to identify subsets of ${2^Q}$ and $2^\alpha$ with formulas of $\mathcal{B}^+(Q)$, $\mathcal{B}^+(\alpha)$ (i.e., we identify a   formula and its set of models). Further, we lift $\delta \colon Q \times \Sigma \mapsto \mathcal{B}(Q)^+$ to $\delta \colon \mathcal{B}^+(Q) \times \Sigma \mapsto \mathcal{B}^+(Q)$ in the canonical way. Finally, given $\varphi \in \mathcal{B}^+(Q)$ and $S \subseteq Q$, we let 
$\varphi[\false / S]$ denote the result of substituting $\false$ for every state of $Q \setminus \alpha$ in $\delta(q, a)$. With these notations, the deterministic Büchi automaton $\mathcal{D}$ equivalent to $\mathcal{A}$ can be described in four lines: $\mathcal{D} = \langle \Sigma, Q', q_0', \delta', \alpha' \rangle$,  where $Q' = \mathcal{B}^+(Q) \times \mathcal{B}^+(\alpha)$, $q_0' = (\theta_0, \false)$, $\alpha' = \{(\theta, \false) \colon \theta \in \mathcal{B}^+(Q) \}$, and
\[\delta'((q, p), a) = \begin{cases}
	(\delta(q, a), \delta(p, a)) &  \text{if $p \not \equiv \false$} \\
	(\delta(q, a), \delta(q, a)[\false / Q \setminus \alpha]) & \text{otherwise.}
\end{cases}\]
\begin{lemma}\label{lem:det:1}
For every $\mathcal{A} \in \aww[2,\R]$ with $n$ states, the deterministic co-Büchi automaton $\mathcal{D}$ defined above satisfies $L(\mathcal{A}) = L(\mathcal{D})$, and has $3^{2^n}$ states. Dually, for every $\mathcal{A}' \in \aww[2,\A]$ with $n'$ states, there exists a deterministic Büchi automaton $\mathcal{D}'$ that has $3^{2^{n'}}$ states and that satisfies $L(\mathcal{A}') = L(\mathcal{D}')$.
\end{lemma}

\begin{proof}
Assume $w$ is accepted by $\mathcal{A}$.
Let $G = (V, E)$ be an accepting run of $\mathcal{A}$ on $w$.
By \Cref{lem:rundag:normalform} there exists an index $k$ such that all levels of $G$ after the $k$-th one are contained in $\alpha$ and have at least one successor.
Therefore, the run $(\Levels_0, \Promising_0), (\Levels_1, \Promising_1) \ldots$ of $\mathcal{D}$ on $w$ satisfies $\Promising_i \neq \emptyset$ for almost all $i$, and so $\mathcal{D}$ accepts.

Assume $w$ is accepted by $\mathcal{D}$. Let $(\Levels_0$, $\Promising_0)$, $(\Levels_1$, $\Promising_1) \ldots$ be the run of $\mathcal{D}$ on $w$.
By definition, there is a $k \geq 0$ such that $\Promising_i \neq \emptyset$ for every $i \geq k$. Choose levels $U_0, U_1, \ldots, U_{k}$ such that
\begin{itemize}
\item $U_k \in \Promising_k$, and
\item for every $1 \leq i \leq k$, choose $U_{i-1}$ as a predecessor of $U_i$ (this is always possible by the definition of $\delta'$).
\end{itemize}
Further, for every $i \geq k$ choose $U_{i+1}$ as a successor of $U_i$. Now, let
$G=(V, E)$ be the graph given by
\begin{itemize}
\item for every $l \geq 0$, $(q, l) \in V$ if{}f $q \in U_l$; and
\item $((q,l), (q', l+1)) \in E$ if{}f $q \in U_l$ and $q' \in S_q$, where $S_q$ is the minimal
model of $\delta(q, w[l])$ used in the definition of successor level.
\end{itemize}
It follows immediately from the definitions that $G$ is an accepting run of $\mathcal{A}$.
The second part is proven by complementing $\mathcal{A}'$, applying the just described construction, and replacing the co-Büchi condition by a Büchi condition.
\end{proof}

\noindent This result leads to a determinisation procedure for $\aww[2]$.

\begin{lemma}\label{lem:det:3}
For every $\mathcal{A} = \langle \Sigma, Q, \theta_0, \delta, \alpha \rangle \in \aww[2]$ with $n = |Q|$ states and $m = |\mathcal{M}_{\theta_0}|$ minimal models of $\theta_0$ there exists an equivalent deterministic Rabin automaton $\mathcal{D}$ with $2^{2^{n + \log_2 m + 2}}$ states and with $m$ Rabin pairs.
\end{lemma}

\begin{proof}
Let $\mathcal{A}=  \langle \Sigma, Q, \theta_0, \delta, \alpha \rangle$. Given $Q' \subseteq Q$, let $\mathcal{A}_{Q'}$ be the AWW[2] obtaining from $\mathcal{A}$ by substituting $\bigwedge_{q \in Q'} q$ for the initial formula $\theta_0$.  We claim that for each minimal model $S \in \mathcal{M}_{\theta_0}$ we can construct a deterministic Rabin automaton (DRW) $\mathcal{D}_S$ with at most $2^{2^{n+2}}$ states and a single Rabin pair, recognising the same language as $\mathcal{A}_{S}$. Let us first see how to construct $\mathcal{D}$, assuming the claim holds. By the claim we have $\lang(\mathcal{A}) = \bigcup_{S \in \mathcal{M}_{\theta_0}} \lang(\mathcal{A}_{S})$.
So we define $\mathcal{D}$ as the union of all the automata $\mathcal{D}_S$. 
Recall that given two  DRWs with $n_1, n_2$ states and $p_1, p_2$ Rabin pairs we can construct a DRW for the union of their languages with $n_1 \times n_2$ states and $n_1 + n_2$ pairs. Since $\theta_0$ has $m$ models, $\mathcal{D}$ has at most $m$ Rabin pairs and $\left(2^{2^{n+2}}\right)^m = 2^{2^{n + \log_2 m + 2}}$ states.

It remains to prove the claim. Partition $S$ into $ S \cap \alpha$ and $S \setminus \alpha$. We have $\mathcal{A}_{S \cap \alpha} \in \aww[2,\A]$ and $\mathcal{A}_{S \setminus \alpha} \in \aww[2,\R]$. By \Cref{lem:det:1} there exists a deterministic Büchi automaton $\mathcal{D}_{S \cap \alpha}$ and a deterministic co-Büchi automaton $\mathcal{D}_{S \setminus \alpha}$ equivalent to  $\mathcal{A}_{S \cap \alpha}$ and
$\mathcal{A}_{S \setminus \alpha}$, respectively, both with at most $3^{2^n}$ states. Intersecting these two automata yields a deterministic Rabin automaton with at most $3^{2^{n+1}} \leq 2^{2^{n+2}}$ states and a single Rabin pair, and we are done.
\end{proof}

\subsection{Translation of LTL to DRW}

We combine the normalisation procedure and the translation of LTL to $\alw$ of the previous section to obtain for every formula of LTL an equivalent DRW of double exponential size. 
Given a formula $\varphi$ we have: 
$\varphi \equiv \bigvee_{\substack{\setmu \subseteq \sfmu(\varphi)\\\setnu \subseteq \sfnu(\varphi)}} \varphi_{\setmu,\setnu}$
\noindent where
\[
\varphi_{\setmu,\setnu} = \left( \flattentwo{\varphi}{\setmu} \wedge \bigwedge_{\psi \in \setmu} \G\F(\evalmu{\psi}{\setnu}) \wedge \bigwedge_{\psi \in \setnu} \F\G(\evalnu{\psi}{\setmu}) \right)
\]
Using the results of Section \ref{subsec:LTLtoAWW2}, we translate each formula $\varphi_{\setmu,\setnu}$ to an $\alw[2]$, and then, applying the determinisation algorithm of  Section \ref{subsec:det}, to a DRW. Finally, using the well-known union operation for DRWs, we obtain a DRW for $\varphi$.

In order to bound the number of states of the final DRW, we first need to determine the number of states of the $\alw$ for each $\varphi_{\setmu,\setnu}$.

\begin{restatable}{lemma}{lemSizeSf}\label{lem:size:sf}
Let $\varphi$ be a formula. For every $M \subseteq \sfmu(\varphi)$ and $N \subseteq \sfnu(\varphi)$, there exists an $\alw[2]$ with $O(|\subf(\varphi)|)$ states that recognises $\lang(\varphi_{\setmu,\setnu})$.
\end{restatable}

\begin{proof}
By \Cref{lem:aww:translation}, some $\alw[2]$ with $O(|\subf(\varphi_{\!\setmu,\setnu}) |)$ states recognises $\lang(\varphi_{\setmu,\setnu})$. So it suffices to show that $|\subf(\varphi_{\setmu,\setnu}) | \in O(|\subf(\varphi)|)$, which follows from these claims, proved in 
\ifarxiv
\Cref{appendix:omitted_proofs}:
\fi
\begin{enumerate}
	\item $|\bigcup \{\subf(\evalnu{\psi}{\setmu}) : \psi \in \subf(\varphi)\} | \leq |\subf(\varphi)|$; 
	\item $|\bigcup \{\subf(\evalmu{\psi}{\setnu}) : \psi \in \subf(\varphi)\} | \leq |\subf(\varphi)|$; 
	\item $|\subf(\flatten{\varphi}{\setmu})| \leq 3 |\subf(\varphi)|$. \qedhere
\end{enumerate}

\end{proof}

\begin{proposition}
Let $\varphi$ be a formula with $n$ proper subformulas. 	There exists a deterministic Rabin automaton recognising $\lang(\varphi)$ with $2^{2^{\mathcal{O}(n)}}$ states and $2^n$ Rabin pairs.
\end{proposition}

\begin{proof}
By \Cref{lem:size:sf} the set $\subf(\varphi_{\setmu,\setnu})$ has at most $O(n)$ elements
 for every $\setmu, \setnu$. Further, due to \Cref{lem:aww:translation} the automaton 
 $\mathcal{A}_{\varphi_{\setmu,\setnu}}$ belongs to $\alw[2]$ and has at most $O(n)$ states. Applying the construction of \Cref{lem:det:3} we obtain a DRW with $2^{2^{O(n)}}$ states and a single Rabin pair. Using the union operation for DRWs we obtain a DRW for $\varphi$ with 
 $\left(2^{2^{O(n)}}\right)^{2^n} = 2^{2^{O(n)}}$ states.
\end{proof}

\begin{remark}
The construction of \Cref{lem:det:1} is close to Miyano and Hayashi's translation of alternating automata to non-de\-ter\-min\-ist\-ic automata \cite{DBLP:journals/tcs/MiyanoH84}, and to Schneider's
translation of $\Sigma_2$ formulas to deterministic co-Büchi automata \cite[p.219]{DBLP:series/txtcs/Schneider04}, all based on the break-point idea.
\end{remark}

\subsection{Determinisation of Lower Classes}

We now determinise $\aww[1]$.
A deterministic automaton is \emph{terminal-accepting} if all states are rejecting except a single accepting sink with a self-loop, and \emph{terminal-rejecting} if all states are accepting except a single rejecting sink with a self-loop. It is easy to see that terminal-accepting and terminal-rejecting deterministic automata are closed under union and intersection. When applied to $\aww[1,\A]$, the construction of \Cref{lem:det:1}, yields automata whose states have a trivial $\Promising$ set (either the empty set or the complete level). Further, the successor of an $\alpha$-level is also an $\alpha$-level. From these observations we easily get:

\begin{corollary}\label{lem:det:4}
Let $\mathcal{A}$ be an automaton with $n$ states.
\begin{itemize}
   \item If $\mathcal{A} \in \aww[1,\R]$ (resp. $\mathcal{A} \in \aww[1,\A]$), then there exists a deterministic terminal-accepting (resp. termi\-nal-re\-jecting) automaton recognising $\lang(\mathcal{A})$ with $2^{2^n}$ states.
    \item If $\mathcal{A} \in \aww[1]$, then there exists deterministic weak automaton recognising $\lang(\mathcal{A})$ with $2^{2^{n + \log_2 |\mathcal{M}_{\theta_0}| + 1}}$ states.
\end{itemize}
\end{corollary}

\subsection{Preliminary Experimental Evaluation}

We expect the LTL-to-DRW translation of this paper to produce automata similar in size (number of states, Rabin pairs) to the translations presented in \cite{DBLP:conf/lics/EsparzaKS18,DBLP:phd/dnb/Sickert19}, which have been implemented using Owl \cite{DBLP:conf/atva/KretinskyMS18} and have been extensively tested.
Indeed, the \enquote{Master Theorem} of \cite{DBLP:conf/lics/EsparzaKS18,DBLP:phd/dnb/Sickert19} characterises the words satisfying a formula $\varphi$ as those for which there exist sets $M$, $N$ of subformulas satisfying three conditions, and so it has the same rough structure as our normal form. 
Further, for each disjunct of our normal form the automata constructions used in \cite{DBLP:conf/lics/EsparzaKS18,DBLP:phd/dnb/Sickert19} and the ones used in this paper are similar. 
Finally, in preliminary experiments we have compared the LTL-to-DRW translations from \cite{DBLP:phd/dnb/Sickert19} and a prototype implementation, without optimisations, of the normalisation procedure of this paper. As benchmark sets we used the \enquote{Dwyer}-patterns \cite{DBLP:conf/fmsp/DwyerAC98}, pre-processed as described in \cite[Ch. 8]{DBLP:phd/dnb/Sickert19}, and the \enquote{Parametrised} formula set from \cite[Ch. 8]{DBLP:phd/dnb/Sickert19}. We observed that on the first set for 60\% of the formulas the number of states of the resulting DRWs was equal, for 17\% the number of states obtained using the construction of this paper was smaller, and for 23\% the number of states was larger. On the second set the ratios were: 76\% equal, 21\% smaller, and 3\% larger. For both sets combined we observed that in 85\% of all 164 cases the difference in number of states was less than or equal to three.

We concluded that the main advantage of our translation is not its performance, but its modularity (it splits the procedure into a normalisation and a simplified translation phase) and its suitability for symbolic automata constructions. 
We leave a detailed experimental comparison and possible integration in Owl \cite{DBLP:conf/atva/KretinskyMS18} (which in particular requires to examine different options for formula and automata simplification, as well as an extensive comparison to existing translations) for future work.

\section{A Hierarchy of Alternating Weak and Very Weak Automata}
\label{sec:hierarchy}

\newcommand{\awwG}{\aww_{\text{G}}}
\newcommand{\alwPS}{\alw_{\text{PS}}}

The expressive power of weak and very weak alternating automata has been studied
by Gurumurthy \textit{et al.} in \cite{DBLP:conf/charme/GurumurthyKSV03} and 
by Pel{\'{a}}nek and Strejcek in \cite{DBLP:conf/wia/PelanekS05},
respectively.  Both papers identify the number of alternations 
between accepting and non-accepting states as an important parameter, and define
a hierarchy of automata classes based on it.  Let $\awwG[k]$ denote the class of $\aww$ with at most $(k-1)$ alternations defined in \cite{DBLP:conf/charme/GurumurthyKSV03}. Similarly, let $\alwPS[k,\A]$
and $\alwPS[k,\R]$ denote the classes of $\alw$ with at most $(k-1)$ alternations and 
accepting or non-accepting initial state, respectively, defined in \cite{DBLP:conf/wia/PelanekS05}. Finally, define $\alwPS[k] = \alwPS[k,\A] \cup \alwPS[k,\R]$\footnote{In \cite{DBLP:conf/wia/PelanekS05} the classes have different names.}.
\Cref{fig:hierarchies} shows the results of \cite{DBLP:conf/charme/GurumurthyKSV03}  
and \cite{DBLP:conf/wia/PelanekS05}. We abuse language, and, for example, write $\Pi_2 = \alwPS[2, \A]$
to denote that the class of languages satisfying formulas in $\Pi_2$ and the class of languages recognized by automata 
in $\alwPS[2, \A]$ coincide.

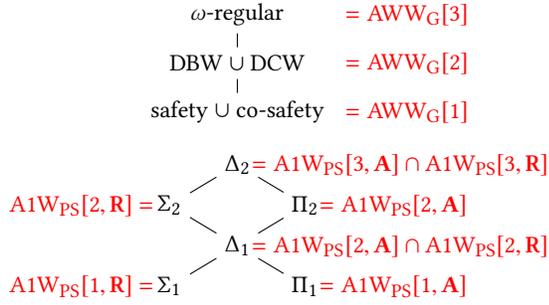
\begin{figure}[bt]
  \small
  \begin{center}
	\begin{tikzpicture}[x=1cm,y=0.60cm,outer sep=0pt,scale=0.9]
	
	\node (m3) at (0, 3.7) {$\omega$-regular}; \node[color=red] (m3l) at (2.5, 3.7) {$= \awwG[3]$};
	\node (m2) at (0, 2.5) {DBW $\cup$ DCW}; \node[color=red] (m2l) at (2.5, 2.5) {$= \awwG[2]$};
	\node (m1) at (0, 1.3) {safety $\cup$ co-safety}; \node[color=red] (m2l) at (2.5, 1.3) {$= \awwG[1]$};
	
	\path[-]
	(m2) edge node{} (m1)
     (m3) edge node{} (m2);

    \node (padding1) at ( 0,0.3) {};
	\node (1) at ( 0, 0) {$\Delta_2$}; \node[color=red] (1l) at (2.4, 0) {$= \alwPS[3,\A] \cap \alwPS[3,\R]$};
	\node (2) at ( 1,-1) {$\Pi_2$}; \node[color=red] (2l) at (2.3, -1) {$= \alwPS[2,\A]$};
	\node (3) at (-1,-1) {$\Sigma_2$}; \node[color=red] (3l) at (-2.3, -1) {$\alwPS[2,\R] = $};
    \node (4) at ( 0,-2) {$\Delta_1$}; \node[color=red] (4l) at (2.4, -2) {$= \alwPS[2,\A] \cap \alwPS[2,\R]$};
	\node (5) at ( 1,-3) {$\Pi_1$}; \node[color=red] (5l) at (2.3, -3) {$= \alwPS[1,\A]$};
	\node (6) at (-1,-3) {$\Sigma_1$}; \node[color=red] (6l) at (-2.3, -3) {$\alwPS[1,\R] = $};

	\path[-]
	(2) edge node{} (1)
    (3) edge node{} (1)
    
    (4) edge node{} (2)
    (4) edge node{} (3)
    
    (5) edge node{} (4)
    (6) edge node{} (4);
	\end{tikzpicture}
  \end{center}
\caption{Expressive power of AWWs after Gurumurthy \textit{et al.} \cite{DBLP:conf/charme/GurumurthyKSV03}, and 
of A1Ws after Pel{\'{a}}nek and Strejcek \cite{DBLP:conf/wia/PelanekS05}.}
\label{fig:hierarchies}
\end{figure}

Unfortunately, the results of \cite{DBLP:conf/charme/GurumurthyKSV03} and \cite{DBLP:conf/wia/PelanekS05} do not \enquote{match}. Due to slight differences in the definitions of height, e.g. the treatment of $\delta(\cdot) = \false$ and $\delta(\cdot) = \true$, the restriction to very weak automata of $\awwG[k]$ does not match any class $\alwPS[k']$ (that is, $\awwG[k] \cap \alw \neq \alwPS[k']$) and, vice versa, extending $\alwPS[k]$ does not yield any $\awwG$ $[k']$. We show that our new definition of height unifies the two hierarchies, yielding the pleasant result shown in \Cref{fig:hierarchies2}. The result follows from \Cref{lem:det:1,lem:det:3,lem:aww:translation}, \Cref{lem:det:4}, and from constructions appearing in \cite{DBLP:conf/ifipTCS/LodingT00,DBLP:conf/charme/GurumurthyKSV03,DBLP:conf/wia/PelanekS05}. 
\ifarxiv A proof sketch is located in \Cref{appendix:omitted_proofs}. \fi

\begin{figure}[bt]
  \small
  \begin{center}
	\begin{tikzpicture}[x=1cm,y=0.60cm,outer sep=0pt,scale=0.9]
	
	\def\d{1.1}
	
	\node (m3) at (0, 4*\d) {$\omega$-regular}; \node[color=red] (m3l) at (1.6, 4*\d) {$= \aww[2]$};
	\node (m2le) at (-0.7, 3*\d) {\dcw}; \node[color=red] (m2l) at (-2.2, 3*\d) {$\aww[2, \R] =$};
	\node (m2ri) at (0.7, 3*\d) {\dbw}; \node[color=red] (m2l) at (2.2, 3*\d) {$= \aww[2, \A]$};
	\node (m1) at (0, 2*\d) {DWW}; \node[color=red] (m2l) at (1.4, 2*\d) {$= \aww[1]$};
	\node (m0le) at (-0.7, 1*\d) {co-safety}; \node[color=red] (m2l) at (-2.5, \d) {$\aww[1,\R] =$};
	\node (m0ri) at (0.7, 1*\d) {safety}; \node[color=red] (m2l) at (2.3, \d) {$= \aww[1,\A]$};
	
	\path[-]
	(m1) edge node{} (m0ri)
     (m1) edge node{} (m0le)
	(m2le) edge node{} (m1)
	(m2ri) edge node{} (m1)
     (m3) edge node{} (m2ri)
     (m3) edge node{} (m2le);

    \node (padding1) at ( 0,0.3) {};
	\node (1) at ( 0, 0) {$\Delta_2$}; \node[color=red] (1l) at (1, 0) {$= \alw[2]$};
	\node (2) at (0.7,-1*\d) {$\Pi_2$}; \node[color=red] (2l) at (1.9, -1*\d) {$= \alw[2,\A]$};
	\node (3) at (-0.7,-1*\d) {$\Sigma_2$}; \node[color=red] (3l) at (-1.9, -1*\d) {$\alw[2,\R] = $};
    \node (4) at ( 0,-2*\d) {$\Delta_1$}; \node[color=red] (4l) at (1, -2*\d) {$= \alw[1]$};
	\node (5) at ( 0.7,-3*\d) {$\Pi_1$}; \node[color=red] (5l) at (1.9, -3*\d) {$= \alw[1,\A]$};
	\node (6) at (-0.7,-3*\d) {$\Sigma_1$}; \node[color=red] (6l) at (-1.9, -3*\d) {$\alw[1,\R] = $};

	\path[-]
	(2) edge node{} (1)
     (3) edge node{} (1)
    
    (4) edge node{} (2)
    (4) edge node{} (3)
    
    (5) edge node{} (4)
    (6) edge node{} (4);

	\end{tikzpicture}
  \end{center}
\caption{Expressive power of AWWs and A1Ws}
\label{fig:hierarchies2}
\end{figure}
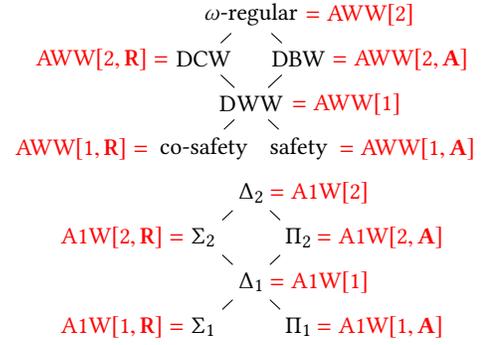

\begin{restatable}{proposition}{propCorrAwwDet}\label{prop:hierarchy:aww}
$\!\aww[2] =$ $\omega$-regular, $\aww[2,\A]$ $=$ $\dbw$, $\aww[2,\R]$ $=$ $\dcw$, $\aww[1]$ $\!=$ \emph{DWW}, $\aww[1,\A]$ $=$ safety, $\aww[1,\R] =$ co-safety, $\alw[1,\R] = \Sigma_1$, $\alw[1,\A] =$ $\Pi_1$, $\alw[1] = \Delta_1$, $\alw[2,\R] = \Sigma_2$, $\alw[2,\A] = \Pi_2$, $\alw[2] = \Delta_2$.
\end{restatable}

\noindent Moreover, our single exponential normalisation procedure for LTL transfers to a single exponential normalisation procedure for $\alw$:
 
\begin{lemma}
Let $\mathcal{A}$ be an $\alw$ with $n$ states over an alphabet with $m$ letters. There exists $\mathcal{A}' \in \alw[2]$ with $2^{\mathcal{O}(nm)}$ states such that $\lang(\mathcal{A}) = \lang(\mathcal{A}')$.
\end{lemma}

\begin{proof}
The translation from $\alw$ to LTL used in \Cref{prop:hierarchy:aww} (an adaption of \cite{DBLP:conf/ifipTCS/LodingT00}) yields a formula $\chi_\mathcal{A}$ with at most $\mathcal{O}(mn)$ proper subformulas. Applying our normalisation procedure to $\chi_\mathcal{A}$ yields an equivalent formula in $\Delta_2$ with at most $2^{\mathcal{O}(mn)}$ proper subformulas (\Cref{lem:size:sf}). 
Applying \Cref{lem:aww:translation} we obtain the postulated automaton $\mathcal{A}'$.
\end{proof} 
\section{Conclusion}

We have presented a purely syntactic normalisation procedure for LTL that transforms a given formula into an equivalent formula in $\Delta_2$, i.e., a formula with at most one alternation
between least- and greatest-fixpoint operators. The procedure has single exponential blow-up,
improving on the prohibitive non-elementary cost of previous constructions. The much better complexity of the new procedure (recall that normalisation procedures for CNF and DNF are
also exponential) makes it attractive for its implementation and use in tools. We have presented 
a first promising application, namely a novel translation from LTL to DRW with double exponential blow-up. Finally, we have shown that the normalisation procedure for LTL can be transferred to a 
normalisation procedure for very weak alternating automata.

Currently we do not know if our normalisation procedure is asymptotically optimal. 
We conjecture that this is the case. 
For the translation of $\aww$ to $\aww[2]$ we also have no further insight, besides the straightforward double exponential upper bound.

\begin{acks}
The authors want to thank Orna Kupferman for the suggestion to examine the expressive power of weak alternating automata and the anonymous reviewers for their helpful comments and remarks.

This work is partly funded by the \grantsponsor{CAVA2}{German Research Foundation (DFG)}{} project \enquote{Verified Model Checkers} (\grantnum{CAVA2}{317422601}) and partly funded by the \grantsponsor{PaVeS}{European Research Council (ERC)}{} under the European Union’s Horizon 2020 research and innovation programme under grant agreement PaVeS (No \grantnum{PaVeS}{787\-367}).
\end{acks}

\bibliography{bibliography}


\begin{thebibliography}{24}


\ifx \showCODEN    \undefined \def \showCODEN     #1{\unskip}     \fi
\ifx \showDOI      \undefined \def \showDOI       #1{#1}\fi
\ifx \showISBNx    \undefined \def \showISBNx     #1{\unskip}     \fi
\ifx \showISBNxiii \undefined \def \showISBNxiii  #1{\unskip}     \fi
\ifx \showISSN     \undefined \def \showISSN      #1{\unskip}     \fi
\ifx \showLCCN     \undefined \def \showLCCN      #1{\unskip}     \fi
\ifx \shownote     \undefined \def \shownote      #1{#1}          \fi
\ifx \showarticletitle \undefined \def \showarticletitle #1{#1}   \fi
\ifx \showURL      \undefined \def \showURL       {\relax}        \fi
\providecommand\bibfield[2]{#2}
\providecommand\bibinfo[2]{#2}
\providecommand\natexlab[1]{#1}
\providecommand\showeprint[2][]{arXiv:#2}

\bibitem[\protect\citeauthoryear{Brunner, Seidl, and Sickert}{Brunner
  et~al\mbox{.}}{2019}]%
        {DBLP:conf/itp/0001SS19}
\bibfield{author}{\bibinfo{person}{Julian Brunner}, \bibinfo{person}{Benedikt
  Seidl}, {and} \bibinfo{person}{Salomon Sickert}.}
  \bibinfo{year}{2019}\natexlab{}.
\newblock \showarticletitle{A Verified and Compositional Translation of {LTL}
  to Deterministic Rabin Automata}. In \bibinfo{booktitle}{\emph{10th
  International Conference on Interactive Theorem Proving, {ITP} 2019,
  September 9-12, 2019, Portland, OR, {USA}}}
  \emph{(\bibinfo{series}{LIPIcs})}, \bibfield{editor}{\bibinfo{person}{John
  Harrison}, \bibinfo{person}{John O'Leary}, {and} \bibinfo{person}{Andrew
  Tolmach}} (Eds.), Vol.~\bibinfo{volume}{141}. \bibinfo{publisher}{Schloss
  Dagstuhl - Leibniz-Zentrum f{\"{u}}r Informatik},
  \bibinfo{pages}{11:1--11:19}.
\newblock
\urldef\tempurl%
\url{https://doi.org/10.4230/LIPIcs.ITP.2019.11}
\showDOI{\tempurl}


\bibitem[\protect\citeauthoryear{{\v{C}}ern{\'{a}} and
  Pel{\'{a}}nek}{{\v{C}}ern{\'{a}} and Pel{\'{a}}nek}{2003}]%
        {DBLP:conf/mfcs/CernaP03}
\bibfield{author}{\bibinfo{person}{Ivana {\v{C}}ern{\'{a}}} {and}
  \bibinfo{person}{Radek Pel{\'{a}}nek}.} \bibinfo{year}{2003}\natexlab{}.
\newblock \showarticletitle{Relating Hierarchy of Temporal Properties to Model
  Checking}. In \bibinfo{booktitle}{\emph{Mathematical Foundations of Computer
  Science 2003, 28th International Symposium, {MFCS} 2003, Bratislava,
  Slovakia, August 25-29, 2003, Proceedings}} \emph{(\bibinfo{series}{Lecture
  Notes in Computer Science})}, \bibfield{editor}{\bibinfo{person}{Branislav
  Rovan} {and} \bibinfo{person}{Peter Vojt{\'{a}}s}} (Eds.),
  Vol.~\bibinfo{volume}{2747}. \bibinfo{publisher}{Springer},
  \bibinfo{pages}{318--327}.
\newblock
\urldef\tempurl%
\url{https://doi.org/10.1007/978-3-540-45138-9\_26}
\showDOI{\tempurl}


\bibitem[\protect\citeauthoryear{Chang, Manna, and Pnueli}{Chang
  et~al\mbox{.}}{1992}]%
        {DBLP:conf/icalp/ChangMP92}
\bibfield{author}{\bibinfo{person}{Edward~Y. Chang}, \bibinfo{person}{Zohar
  Manna}, {and} \bibinfo{person}{Amir Pnueli}.}
  \bibinfo{year}{1992}\natexlab{}.
\newblock \showarticletitle{Characterization of Temporal Property Classes}. In
  \bibinfo{booktitle}{\emph{Automata, Languages and Programming, 19th
  International Colloquium, ICALP92, Vienna, Austria, July 13-17, 1992,
  Proceedings}} \emph{(\bibinfo{series}{Lecture Notes in Computer Science})},
  \bibfield{editor}{\bibinfo{person}{Werner Kuich}} (Ed.),
  Vol.~\bibinfo{volume}{623}. \bibinfo{publisher}{Springer},
  \bibinfo{pages}{474--486}.
\newblock
\urldef\tempurl%
\url{https://doi.org/10.1007/3-540-55719-9\_97}
\showDOI{\tempurl}


\bibitem[\protect\citeauthoryear{Dwyer, Avrunin, and Corbett}{Dwyer
  et~al\mbox{.}}{1998}]%
        {DBLP:conf/fmsp/DwyerAC98}
\bibfield{author}{\bibinfo{person}{Matthew~B. Dwyer},
  \bibinfo{person}{George~S. Avrunin}, {and} \bibinfo{person}{James~C.
  Corbett}.} \bibinfo{year}{1998}\natexlab{}.
\newblock \showarticletitle{Property specification patterns for finite-state
  verification}. In \bibinfo{booktitle}{\emph{{FMSP}}}. \bibinfo{pages}{7--15}.
\newblock
\urldef\tempurl%
\url{https://doi.org/10.1145/298595.298598}
\showDOI{\tempurl}


\bibitem[\protect\citeauthoryear{Esparza, Kret{\'{\i}}nsk{\'{y}}, and
  Sickert}{Esparza et~al\mbox{.}}{2018}]%
        {DBLP:conf/lics/EsparzaKS18}
\bibfield{author}{\bibinfo{person}{Javier Esparza}, \bibinfo{person}{Jan
  Kret{\'{\i}}nsk{\'{y}}}, {and} \bibinfo{person}{Salomon Sickert}.}
  \bibinfo{year}{2018}\natexlab{}.
\newblock \showarticletitle{One Theorem to Rule Them All: {A} Unified
  Translation of {LTL} into {\(\omega\)}-Automata}. In
  \bibinfo{booktitle}{\emph{Proceedings of the 33rd Annual {ACM/IEEE} Symposium
  on Logic in Computer Science, {LICS} 2018, Oxford, UK, July 09-12, 2018}},
  \bibfield{editor}{\bibinfo{person}{Anuj Dawar} {and} \bibinfo{person}{Erich
  Gr{\"{a}}del}} (Eds.). \bibinfo{publisher}{{ACM}}, \bibinfo{pages}{384--393}.
\newblock
\urldef\tempurl%
\url{https://doi.org/10.1145/3209108.3209161}
\showDOI{\tempurl}


\bibitem[\protect\citeauthoryear{Gurumurthy, Kupferman, Somenzi, and
  Vardi}{Gurumurthy et~al\mbox{.}}{2003}]%
        {DBLP:conf/charme/GurumurthyKSV03}
\bibfield{author}{\bibinfo{person}{Sankar Gurumurthy}, \bibinfo{person}{Orna
  Kupferman}, \bibinfo{person}{Fabio Somenzi}, {and} \bibinfo{person}{Moshe~Y.
  Vardi}.} \bibinfo{year}{2003}\natexlab{}.
\newblock \showarticletitle{On Complementing Nondeterministic B{\"{u}}chi
  Automata}. In \bibinfo{booktitle}{\emph{Correct Hardware Design and
  Verification Methods, 12th {IFIP} {WG} 10.5 Advanced Research Working
  Conference, {CHARME} 2003, L'Aquila, Italy, October 21-24, 2003,
  Proceedings}} \emph{(\bibinfo{series}{Lecture Notes in Computer Science})},
  \bibfield{editor}{\bibinfo{person}{Daniel Geist} {and}
  \bibinfo{person}{Enrico Tronci}} (Eds.), Vol.~\bibinfo{volume}{2860}.
  \bibinfo{publisher}{Springer}, \bibinfo{pages}{96--110}.
\newblock
\urldef\tempurl%
\url{https://doi.org/10.1007/978-3-540-39724-3\_10}
\showDOI{\tempurl}


\bibitem[\protect\citeauthoryear{Kret{\'{\i}}nsk{\'{y}}, Meggendorfer, and
  Sickert}{Kret{\'{\i}}nsk{\'{y}} et~al\mbox{.}}{2018}]%
        {DBLP:conf/atva/KretinskyMS18}
\bibfield{author}{\bibinfo{person}{Jan Kret{\'{\i}}nsk{\'{y}}},
  \bibinfo{person}{Tobias Meggendorfer}, {and} \bibinfo{person}{Salomon
  Sickert}.} \bibinfo{year}{2018}\natexlab{}.
\newblock \showarticletitle{Owl: {A} Library for {\(\omega\)}-Words, Automata,
  and {LTL}}. In \bibinfo{booktitle}{\emph{Automated Technology for
  Verification and Analysis - 16th International Symposium, {ATVA} 2018, Los
  Angeles, CA, USA, October 7-10, 2018, Proceedings}}
  \emph{(\bibinfo{series}{Lecture Notes in Computer Science})},
  \bibfield{editor}{\bibinfo{person}{Shuvendu~K. Lahiri} {and}
  \bibinfo{person}{Chao Wang}} (Eds.), Vol.~\bibinfo{volume}{11138}.
  \bibinfo{publisher}{Springer}, \bibinfo{pages}{543--550}.
\newblock
\urldef\tempurl%
\url{https://doi.org/10.1007/978-3-030-01090-4\_34}
\showDOI{\tempurl}


\bibitem[\protect\citeauthoryear{Lichtenstein, Pnueli, and Zuck}{Lichtenstein
  et~al\mbox{.}}{1985}]%
        {DBLP:conf/lop/LichtensteinPZ85}
\bibfield{author}{\bibinfo{person}{Orna Lichtenstein}, \bibinfo{person}{Amir
  Pnueli}, {and} \bibinfo{person}{Lenore~D. Zuck}.}
  \bibinfo{year}{1985}\natexlab{}.
\newblock \showarticletitle{The Glory of the Past}. In
  \bibinfo{booktitle}{\emph{Logics of Programs, Conference, Brooklyn College,
  New York, NY, USA, June 17-19, 1985, Proceedings}}
  \emph{(\bibinfo{series}{Lecture Notes in Computer Science})},
  \bibfield{editor}{\bibinfo{person}{Rohit Parikh}} (Ed.),
  Vol.~\bibinfo{volume}{193}. \bibinfo{publisher}{Springer},
  \bibinfo{pages}{196--218}.
\newblock
\urldef\tempurl%
\url{https://doi.org/10.1007/3-540-15648-8\_16}
\showDOI{\tempurl}


\bibitem[\protect\citeauthoryear{L{\"{o}}ding and Thomas}{L{\"{o}}ding and
  Thomas}{2000}]%
        {DBLP:conf/ifipTCS/LodingT00}
\bibfield{author}{\bibinfo{person}{Christof L{\"{o}}ding} {and}
  \bibinfo{person}{Wolfgang Thomas}.} \bibinfo{year}{2000}\natexlab{}.
\newblock \showarticletitle{Alternating Automata and Logics over Infinite
  Words}. In \bibinfo{booktitle}{\emph{Theoretical Computer Science, Exploring
  New Frontiers of Theoretical Informatics, International Conference {IFIP}
  {TCS} 2000, Sendai, Japan, August 17-19, 2000, Proceedings}}
  \emph{(\bibinfo{series}{Lecture Notes in Computer Science})},
  \bibfield{editor}{\bibinfo{person}{Jan van Leeuwen}, \bibinfo{person}{Osamu
  Watanabe}, \bibinfo{person}{Masami Hagiya}, \bibinfo{person}{Peter~D.
  Mosses}, {and} \bibinfo{person}{Takayasu Ito}} (Eds.),
  Vol.~\bibinfo{volume}{1872}. \bibinfo{publisher}{Springer},
  \bibinfo{pages}{521--535}.
\newblock
\urldef\tempurl%
\url{https://doi.org/10.1007/3-540-44929-9\_36}
\showDOI{\tempurl}


\bibitem[\protect\citeauthoryear{Maler and Pnueli}{Maler and Pnueli}{1990}]%
        {DBLP:conf/focs/MalerP90}
\bibfield{author}{\bibinfo{person}{Oded Maler} {and} \bibinfo{person}{Amir
  Pnueli}.} \bibinfo{year}{1990}\natexlab{}.
\newblock \showarticletitle{Tight Bounds on the Complexity of Cascaded
  Decomposition of Automata}. In \bibinfo{booktitle}{\emph{31st Annual
  Symposium on Foundations of Computer Science, St. Louis, Missouri, USA,
  October 22-24, 1990, Volume {II}}}. \bibinfo{publisher}{{IEEE} Computer
  Society}, \bibinfo{pages}{672--682}.
\newblock
\urldef\tempurl%
\url{https://doi.org/10.1109/FSCS.1990.89589}
\showDOI{\tempurl}


\bibitem[\protect\citeauthoryear{Manna and Pnueli}{Manna and Pnueli}{1990}]%
        {DBLP:conf/podc/MannaP89}
\bibfield{author}{\bibinfo{person}{Zohar Manna} {and} \bibinfo{person}{Amir
  Pnueli}.} \bibinfo{year}{1990}\natexlab{}.
\newblock \showarticletitle{A Hierarchy of Temporal Properties}. In
  \bibinfo{booktitle}{\emph{Proceedings of the Ninth Annual {ACM} Symposium on
  Principles of Distributed Computing, Quebec City, Quebec, Canada, August
  22-24, 1990}}, \bibfield{editor}{\bibinfo{person}{Cynthia Dwork}} (Ed.).
  \bibinfo{publisher}{{ACM}}, \bibinfo{pages}{377--410}.
\newblock
\urldef\tempurl%
\url{https://doi.org/10.1145/93385.93442}
\showDOI{\tempurl}


\bibitem[\protect\citeauthoryear{Manna and Pnueli}{Manna and Pnueli}{1992}]%
        {DBLP:books/daglib/0077033}
\bibfield{author}{\bibinfo{person}{Zohar Manna} {and} \bibinfo{person}{Amir
  Pnueli}.} \bibinfo{year}{1992}\natexlab{}.
\newblock \bibinfo{booktitle}{\emph{The temporal logic of reactive and
  concurrent systems - specification}}.
\newblock \bibinfo{publisher}{Springer}.
\newblock


\bibitem[\protect\citeauthoryear{Miyano and Hayashi}{Miyano and
  Hayashi}{1984}]%
        {DBLP:journals/tcs/MiyanoH84}
\bibfield{author}{\bibinfo{person}{Satoru Miyano} {and}
  \bibinfo{person}{Takeshi Hayashi}.} \bibinfo{year}{1984}\natexlab{}.
\newblock \showarticletitle{Alternating Finite Automata on omega-Words}.
\newblock \bibinfo{journal}{\emph{Theor. Comput. Sci.}}  \bibinfo{volume}{32}
  (\bibinfo{year}{1984}), \bibinfo{pages}{321--330}.
\newblock
\urldef\tempurl%
\url{https://doi.org/10.1016/0304-3975(84)90049-5}
\showDOI{\tempurl}


\bibitem[\protect\citeauthoryear{Muller, Saoudi, and Schupp}{Muller
  et~al\mbox{.}}{1988}]%
        {DBLP:conf/lics/MullerSS88}
\bibfield{author}{\bibinfo{person}{David~E. Muller}, \bibinfo{person}{Ahmed
  Saoudi}, {and} \bibinfo{person}{Paul~E. Schupp}.}
  \bibinfo{year}{1988}\natexlab{}.
\newblock \showarticletitle{Weak Alternating Automata Give a Simple Explanation
  of Why Most Temporal and Dynamic Logics are Decidable in Exponential Time}.
  In \bibinfo{booktitle}{\emph{Proceedings of the Third Annual Symposium on
  Logic in Computer Science {(LICS} '88), Edinburgh, Scotland, UK, July 5-8,
  1988}}. \bibinfo{publisher}{{IEEE} Computer Society},
  \bibinfo{pages}{422--427}.
\newblock
\urldef\tempurl%
\url{https://doi.org/10.1109/LICS.1988.5139}
\showDOI{\tempurl}


\bibitem[\protect\citeauthoryear{Nipkow, Paulson, and Wenzel}{Nipkow
  et~al\mbox{.}}{2002}]%
        {DBLP:books/sp/NipkowPW02}
\bibfield{author}{\bibinfo{person}{Tobias Nipkow}, \bibinfo{person}{Lawrence~C.
  Paulson}, {and} \bibinfo{person}{Markus Wenzel}.}
  \bibinfo{year}{2002}\natexlab{}.
\newblock \bibinfo{booktitle}{\emph{Isabelle/HOL - {A} Proof Assistant for
  Higher-Order Logic}}. \bibinfo{series}{Lecture Notes in Computer Science},
  Vol.~\bibinfo{volume}{2283}.
\newblock \bibinfo{publisher}{Springer}.
\newblock
\showISBNx{3-540-43376-7}
\urldef\tempurl%
\url{https://doi.org/10.1007/3-540-45949-9}
\showDOI{\tempurl}


\bibitem[\protect\citeauthoryear{Pel{\'{a}}nek and Strejcek}{Pel{\'{a}}nek and
  Strejcek}{2005}]%
        {DBLP:conf/wia/PelanekS05}
\bibfield{author}{\bibinfo{person}{Radek Pel{\'{a}}nek} {and}
  \bibinfo{person}{Jan Strejcek}.} \bibinfo{year}{2005}\natexlab{}.
\newblock \showarticletitle{Deeper Connections Between {LTL} and Alternating
  Automata}. In \bibinfo{booktitle}{\emph{Implementation and Application of
  Automata, 10th International Conference, {CIAA} 2005, Sophia Antipolis,
  France, June 27-29, 2005, Revised Selected Papers}}
  \emph{(\bibinfo{series}{Lecture Notes in Computer Science})},
  \bibfield{editor}{\bibinfo{person}{Jacques Farr{\'{e}}},
  \bibinfo{person}{Igor Litovsky}, {and} \bibinfo{person}{Sylvain Schmitz}}
  (Eds.), Vol.~\bibinfo{volume}{3845}. \bibinfo{publisher}{Springer},
  \bibinfo{pages}{238--249}.
\newblock
\urldef\tempurl%
\url{https://doi.org/10.1007/11605157\_20}
\showDOI{\tempurl}


\bibitem[\protect\citeauthoryear{Piterman and Pnueli}{Piterman and
  Pnueli}{2018}]%
        {PitermanP18}
\bibfield{author}{\bibinfo{person}{Nir Piterman} {and} \bibinfo{person}{Amir
  Pnueli}.} \bibinfo{year}{2018}\natexlab{}.
\newblock \showarticletitle{Temporal Logic and Fair Discrete Systems}.
\newblock In \bibinfo{booktitle}{\emph{Handbook of Model Checking}},
  \bibfield{editor}{\bibinfo{person}{Edmund~M. Clarke},
  \bibinfo{person}{Thomas~A. Henzinger}, \bibinfo{person}{Helmut Veith}, {and}
  \bibinfo{person}{Roderick Bloem}} (Eds.). \bibinfo{publisher}{Springer},
  \bibinfo{pages}{27--73}.
\newblock
\urldef\tempurl%
\url{https://doi.org/10.1007/978-3-319-10575-8\_2}
\showDOI{\tempurl}


\bibitem[\protect\citeauthoryear{Schneider}{Schneider}{2004}]%
        {DBLP:series/txtcs/Schneider04}
\bibfield{author}{\bibinfo{person}{Klaus Schneider}.}
  \bibinfo{year}{2004}\natexlab{}.
\newblock \bibinfo{booktitle}{\emph{Verification of Reactive Systems - Formal
  Methods and Algorithms}}.
\newblock \bibinfo{publisher}{Springer}.
\newblock
\showISBNx{978-3-642-05555-3}
\urldef\tempurl%
\url{https://doi.org/10.1007/978-3-662-10778-2}
\showDOI{\tempurl}


\bibitem[\protect\citeauthoryear{Seidl and Sickert}{Seidl and Sickert}{2019}]%
        {DBLP:journals/afp/SeidlS19}
\bibfield{author}{\bibinfo{person}{Benedikt Seidl} {and}
  \bibinfo{person}{Salomon Sickert}.} \bibinfo{year}{2019}\natexlab{}.
\newblock \showarticletitle{A Compositional and Unified Translation of {LTL}
  into {\(\omega\)}-Automata}.
\newblock \bibinfo{journal}{\emph{Archive of Formal Proofs}}
  \bibinfo{volume}{2019} (\bibinfo{year}{2019}).
\newblock
\urldef\tempurl%
\url{https://www.isa-afp.org/entries/LTL\_Master\_Theorem.html}
\showURL{%
\tempurl}


\bibitem[\protect\citeauthoryear{Sickert}{Sickert}{2016}]%
        {DBLP:journals/afp/Sickert16}
\bibfield{author}{\bibinfo{person}{Salomon Sickert}.}
  \bibinfo{year}{2016}\natexlab{}.
\newblock \showarticletitle{Linear Temporal Logic}.
\newblock \bibinfo{journal}{\emph{Archive of Formal Proofs}}
  \bibinfo{volume}{2016} (\bibinfo{year}{2016}).
\newblock
\urldef\tempurl%
\url{https://www.isa-afp.org/entries/LTL.shtml}
\showURL{%
\tempurl}


\bibitem[\protect\citeauthoryear{Sickert}{Sickert}{2019}]%
        {DBLP:phd/dnb/Sickert19}
\bibfield{author}{\bibinfo{person}{Salomon Sickert}.}
  \bibinfo{year}{2019}\natexlab{}.
\newblock \emph{\bibinfo{title}{A Unified Translation of Linear Temporal Logic
  to {\(\omega\)}-Automata}}.
\newblock \bibinfo{thesistype}{Ph.D. Dissertation}. \bibinfo{school}{Technical
  University of Munich, Germany}.
\newblock
\urldef\tempurl%
\url{http://nbn-resolving.de/urn:nbn:de:bvb:91-diss-20190801-1484932-1-4}
\showURL{%
\tempurl}


\bibitem[\protect\citeauthoryear{Sickert}{Sickert}{2020}]%
        {XXXX:journals/afp/Sickert20}
\bibfield{author}{\bibinfo{person}{Salomon Sickert}.}
  \bibinfo{year}{2020}\natexlab{}.
\newblock \showarticletitle{An Efficient Normalisation Procedure for Linear
  Temporal Logic: Isabelle/HOL Formalisation}.
\newblock \bibinfo{journal}{\emph{Archive of Formal Proofs}}
  \bibinfo{volume}{2020} (\bibinfo{year}{2020}).
\newblock
\urldef\tempurl%
\url{https://www.isa-afp.org/entries/LTL\_Normal\_Form.html}
\showURL{%
\tempurl}


\bibitem[\protect\citeauthoryear{Vardi}{Vardi}{1994}]%
        {DBLP:conf/tacs/Vardi94}
\bibfield{author}{\bibinfo{person}{Moshe~Y. Vardi}.}
  \bibinfo{year}{1994}\natexlab{}.
\newblock \showarticletitle{Nontraditional Applications of Automata Theory}. In
  \bibinfo{booktitle}{\emph{Theoretical Aspects of Computer Software,
  International Conference {TACS} '94, Sendai, Japan, April 19-22, 1994,
  Proceedings}} \emph{(\bibinfo{series}{Lecture Notes in Computer Science})},
  \bibfield{editor}{\bibinfo{person}{Masami Hagiya} {and}
  \bibinfo{person}{John~C. Mitchell}} (Eds.), Vol.~\bibinfo{volume}{789}.
  \bibinfo{publisher}{Springer}, \bibinfo{pages}{575--597}.
\newblock
\urldef\tempurl%
\url{https://doi.org/10.1007/3-540-57887-0\_116}
\showDOI{\tempurl}


\bibitem[\protect\citeauthoryear{Zuck}{Zuck}{1986}]%
        {XXXX:phd/Zuck86}
\bibfield{author}{\bibinfo{person}{Lenore~D. Zuck}.}
  \bibinfo{year}{1986}\natexlab{}.
\newblock \emph{\bibinfo{title}{{Past Temporal Logic}}}.
\newblock \bibinfo{thesistype}{Ph.D. Dissertation}. \bibinfo{school}{The
  Weizmann Institute of Science, Israel}.
\newblock


\end{thebibliography}

\ifarxiv
\newpage
\appendix
\section{Proofs for the Lemmas from \cite{DBLP:conf/lics/EsparzaKS18,DBLP:phd/dnb/Sickert19}}
\label{sec:updated_proofs}

Since we had to change the notations of \cite{DBLP:conf/lics/EsparzaKS18,DBLP:phd/dnb/Sickert19}, we include for convenience proofs in the new notation.

\lemEvalnu*

\begin{proof}
All parts are proved by a straightforward structural induction on $\varphi$. Here we only present two cases of the induction for (1) and (2). (3) then follows from (1) and (2).

(1) Assume $\setF_w^{\,\varphi} \subseteq \setmu$. Then $\setF_{w_i}^{\,\varphi} \subseteq \setmu$ for all $i \geq 0$. We prove the following stronger statement via structural induction on $\varphi$. We consider one representative of the \enquote{interesting} cases and one of the \enquote{straightforward} cases:
\[\forall i.\, (\, (w_i \models \varphi) \implies (w_i \models \evalnu{\varphi}{\setmu}) \, )\]

\noindent Case $\varphi = \psi_1 \U \psi_2$. Let $i \geq 0$ and assume $w_i \models \psi_1 \U \psi_2$. Then $\psi_1 \U \psi_2 \in \setF_{w_i}^{\,\varphi}$ and so $\varphi \in \setmu$. We prove $w_i \models \evalnu{(\psi_1 \U \psi_2)}{\setmu}$:
\begin{align*}
 	      & \SetToWidest{w_i \models \psi_1 \U \psi_2} & \SetToWidestTwo{}  \\
\implies  & w_i \models \psi_1 \W \psi_2 \\
\iff      & \forall j. \, w_{i+j} \models \psi_1 \vee \exists k \leq j.\, w_{i+k} \models \psi_2 \\
\implies  & \forall j. \, w_{i+j} \models \evalnu{\psi_1}{\setmu} \vee \exists k \leq j.\, w_{i+k} \models \evalnu{\psi_2}{\setmu} & \text{(I.H.)} \\
\implies  & w_i \models (\evalnu{\psi_1}{\setmu}) \W (\evalnu{\psi_2}{\setmu}) \\
\iff      & w_i \models \evalnu{(\psi_1 \U \psi_2)}{\setmu} 
\end{align*}

\noindent Case $\varphi = \psi_1 \vee \psi_2$. Let $i \geq 0$ and assume $w_i \models \psi_1 \vee \psi_2$:
\begin{align*}
	     & \SetToWidest{w_i \models \psi_1 \vee \psi_2} & \SetToWidestTwo{} \\
\iff     & w_i \models \psi_1 \vee w_i \models \psi_2 \\
\implies & w_i \models \evalnu{\psi_1}{\setmu} \vee w_i \models \evalnu{\psi_2}{\setmu} & \text{(I.H.)} \\
\iff     & w_i \models \evalnu{(\psi_1 \vee \psi_2)}{\setmu}                       
\end{align*}

(2) Assume $\setmu \subseteq \setGF_w^{\;\varphi}$. Then $\setmu \subseteq \setGF_{w_i}^{\;\varphi}$ for all $i \geq 0$. We prove the following stronger statement via structural induction on $\varphi$:
\[\forall i. \, (\, (w_i \models \evalnu{\varphi}{\setmu}) \implies (w_i \models \varphi) \,)\]

\noindent Case $\varphi = \psi_1 \U \psi_2$. If $\varphi \notin \setmu$, then by definition $\evalnu{\varphi}{\setmu} = \false$. So $w_i \not \models \evalnu{\varphi}{\setmu} = \false$ for all $i$ and thus the implication $(w_i \models \evalnu{\varphi}{\setmu}) \implies (w_i \models \varphi)$ holds for every $i \geq 0$. Assume now $\varphi \in \setmu$. Since $\setmu \subseteq \setGF_w^{\;\varphi}$ we have $w_i \models \G \F \varphi$ and so in particular $w_i \models \F \psi_2$. To prove the implication assume $w_i \models \evalnu{(\psi_1 \U \psi_2)}{\setmu}$ for an arbitrary fixed $i$. We show $w_i \models \psi_1 \U \psi_2$:
\begin{align*}
	     & \SetToWidest{w_i \models \evalnu{(\psi_1 \U \psi_2)}{\setmu}} & \SetToWidestTwo{} \\
\iff     & w_i \models (\evalnu{\psi_1}{\setmu}) \W (\evalnu{\psi_2}{\setmu}) \\
\iff     & \forall j. \, w_{i+j} \models \evalnu{\psi_1}{\setmu} \vee \exists k \leq j.\, w_{i+k} \models \evalnu{\psi_2}{\setmu} \\
\implies & \forall j. \, w_{i+j} \models \psi_1 \vee \exists k \leq j.\, w_{i+k} \models \psi_2 & \text{(I.H)} \\
\iff     & w_i \models \psi_1 \W \psi_2 \\
\iff     & w_i \models \psi_1 \U \psi_2 
\end{align*}

\noindent Case $\varphi = \psi_1 \vee \psi_2$. Let $i \geq 0$ arbitrary and assume $w_i \models \psi_1 \vee \psi_2$. We have:
\begin{align*}
	     & \SetToWidest{w_i \models \evalnu{(\psi_1 \vee \psi_2)}{\setmu}} & \SetToWidestTwo{} \\
\iff     & w_i \models \evalnu{\psi_1}{\setmu} \vee (w_i \models \evalnu{\psi_2}{\setmu} \\
\implies & w_i \models \psi_1 \vee w_i \models \psi_2 & \text{(I.H.)} \\
\iff     & w_i \models \psi_1 \vee \psi_2
\end{align*}
\end{proof}

\lemEvalmu*

\begin{proof}
All parts are proved by a straightforward structural induction on $\varphi$. Here we only present two cases of the induction for (1) and (2). (3) then follows from (1) and (2).

(1) Assume $\,\setFG_w^{\,\varphi} \subseteq \setnu$. Then $\setFG_{w_i}^{\,\varphi} \subseteq \setnu$ for all $i$. We prove the following stronger statement via structural induction on $\varphi$:
\[\forall i.\,(\,(w_i \models \varphi) \implies (w_i \models \evalmu{\varphi}{\setnu})\,)\]

\noindent Case $\varphi = \psi_1 \W \psi_2$. Let $i \geq 0$ arbitrary and assume $w_i \models \varphi$. If $\varphi \in \setnu$ then  $\evalmu{\varphi}{\setnu} = \true$ and so $w_i \models \evalmu{\varphi}{\setnu} $ trivially holds. Assume now $\varphi \notin \setnu$. Since $\setFG_{w_i}^{\,\varphi} \subseteq \setnu$ we have $w_i \not \models \F\G \varphi$ and so in particular $w_i \not \models \G\psi_1$. We prove $w_i \models \evalmu{(\psi_1 \W \psi_2)}{\setnu}$:
\begin{align*}
 	     & \SetToWidest{w_i \models \psi_1 \W \psi_2} & \SetToWidestTwo{} \\
\iff     & w_i \models \psi_1 \U \psi_2 \\
\iff     & \exists j. \, w_{i+j} \models \psi_2 \wedge \forall k < j.\, w_{i+k} \models \psi_1 \\
\implies & \exists j. \, w_{i+j} \models \evalmu{\psi_2}{\setnu} \wedge \forall k < j.\, w_{i+k} \models \evalmu{\psi_1}{\setnu} & \text{(I.H.)} \\
\iff     & w_i \models (\evalmu{\psi_1}{\setnu}) \U (\evalmu{\psi_2}{\setnu}) \\
\iff     & w_i \models \evalmu{(\psi_1 \W \psi_2)}{\setnu}
\end{align*}

\noindent Case $\varphi = \psi_1 \vee \psi_2$. Let $i \geq 0$ arbitrary and assume $w_i \models \psi_1 \vee \psi_2$. We have:
\begin{align*}
	     & \SetToWidest{w_i \models \psi_1 \vee \psi_2} & \SetToWidestTwo{} \\
\iff     & w_i \models \psi_1 \vee w_i \models \psi_2 \\
\implies & w_i \models \evalmu{\psi_1}{\setnu} \vee w_i \models \evalmu{\psi_2}{\setnu} & \text{(I.H.)} \\
\iff     & w_i \models \evalmu{(\psi_1 \vee \psi_2)}{\setnu}
\end{align*}

(2) Assume $\setnu \subseteq \setG_w^{\varphi}$. Then $\setnu \subseteq \setG_{w_i}^\varphi$ for all $i$. We prove the following stronger statement via structural induction on $\varphi$:
\[\forall i. \, (\,(w_i \models \evalmu{\varphi}{\setnu}) \implies (w_i \models \varphi)\,)\]

\noindent Case $\varphi = \psi_1 \W \psi_2$. If $\varphi \in \setnu$, then since $\setnu \subseteq \setG_w^{\varphi}$ we have $w_i \models \G \varphi$ and so $w_i \models \varphi$. Assume now that $\varphi \notin \setnu$ and $w_i \models \evalmu{(\psi_1 \W \psi_2)}{\setnu}$ for an arbitrary fixed $i$. We prove $w_i \models \psi_1 \W \psi_2$:
\begin{align*}
	     & \SetToWidest{w_i \models \evalmu{(\psi_1 \W \psi_2)}{\setnu}} & \SetToWidestTwo{} \\
\iff     & w_i \models (\evalmu{\psi_1}{\setnu}) \U (\evalmu{\psi_2}{\setnu}) \\
\iff     & \exists j. \, w_{i+j} \models \evalmu{\psi_2}{\setnu} \wedge \forall k < j.\, w_{i+k} \models \evalmu{\psi_1}{\setnu} \\
\implies & \exists j. \, w_{i+j} \models \psi_2 \wedge \forall k < j.\, w_{i+k} \models \psi_1 & \text{(I.H.)} \\
\iff     & w_i \models \psi_1 \U \psi_2 \\
\implies & w_i \models \psi_1 \W \psi_2
\end{align*}

\noindent Case $\varphi = \psi_1 \vee \psi_2$. We derive in a straightforward manner for an arbitrary and fixed $i$:
\begin{align*}
	     & \SetToWidest{w_i \models \evalmu{(\psi_1 \vee \psi_2)}{\setnu}} & \SetToWidestTwo{} \\
\iff     & w_i \models \evalmu{\psi_1}{\setnu} \vee w_i \models \evalmu{\psi_2}{\setnu} \\
\implies & w_i \models \psi_1 \vee w_i \models \psi_2 & \text{(I.H.)} \\
\iff     & w_i \models \psi_1 \vee \psi_2
\end{align*}
\end{proof}

\lemInduction*

\begin{proof}
Let us first focus on part (1) and then move to part (2).

(1) Let $\psi \in \setGF_w^{\;\varphi}$. We have $w \models \G\F\psi$, and so $w_i \models \psi$ for infinitely many $i \geq 0$. Since $\setFG_{w_i}^{\,\varphi} = \setFG_w^{\,\varphi}$ for every $i \geq 0$, \Cref{lem:evalmu}.1 can be applied to $w_i$, $\setFG_{w_i}^{\,\varphi}$, and $\psi$. This yields $w_i \models \evalmu{\psi}{\setFG_w^{\,\varphi}}$ for infinitely many $i \geq 0$ and thus $w \models \G\F(\evalmu{\psi}{\setFG_w^{\,\varphi}})$.

Let $\psi \in \setFG_w^{\,\varphi}$. Since $w_i \models \F\G\psi$, there is an index $j$ such that  $w_{j+k} \models \psi$ for every $k \geq 0$. The index $j$ can be chosen so that it also satisfies $\setGF_{w}^{\;\varphi} = \setF_{w_{j+k}}^{\,\varphi} = \setGF_{w_{j+k}}^{\;\varphi}$ for every $k\geq 0$. So \Cref{lem:evalnu}.1 can be applied to $\setF_{w_{j+k}}^{\,\varphi}$, $w_{j+k}$, and $\psi$. This yields $w_{j+k} \models \evalnu{\psi}{\setGF_w^{\;\varphi}}$ for every $k \geq 0$ and thus $w \models \F\G(\evalnu{\psi}{\setGF_w^{\;\varphi}})$.

(2) Let $\setmu \subseteq \sfmu(\varphi)$ and $\setnu \subseteq \sfnu(\varphi)$. Observe that $\setmu \cap \setnu = \emptyset$. Let $n \coloneqq |\setmu \cup \setnu|$. Let $\psi_1, \ldots, \psi_n$ be an enumeration of $\setmu \cup \setnu$ compatible with the subformula order, i.e., if $\psi_i$ is a subformula of $\psi_j$, then $i \leq j$. Let $(\setmu_0, \setnu_0), (\setmu_1, \setnu_1), \ldots, (\setmu_{n}, \setnu_{n})$ be the unique sequence of pairs satisfying:

\begin{itemize}
	\item $(\setmu_0, \setnu_0) = (\emptyset, \emptyset)$ and  $(\setmu_n, \setnu_n) = (\setmu, \setnu)$.
	\item For every $0 < i \leq n$, if $\psi_i \in \setmu$ then $\setmu_{i} \setminus \setmu_{i-1} = \{ \psi_i \}$ and $\setnu_i = \setnu_{i-1}$, and if $\psi_i \in \setnu$, then $\setmu_i = \setmu_{i-1}$ and $\setnu_{i} \setminus \setnu_{i-1} = \{ \psi_i \}$.
\end{itemize}

We prove $\setmu_i \subseteq \setGF_w^{\;\varphi}$ and $\setnu_i \subseteq \setFG_w^{\,\varphi}$ for every $0 \leq i \leq n$ by induction on $i$. For $i=0$ the result follows immediately from $\setmu_0 = \emptyset= \setnu_0$. For $i > 0$ we consider two cases:

\begin{itemize}
	\item $\psi_{i} \in \setnu$, i.e., $\setmu_i = \setmu_{i-1}$ and $\setnu_i \setminus \setnu_{i-1} = \{ \psi_i \}$.

	By induction hypothesis and $\setmu_{i} = \setmu_{i-1}$ we have $\setmu_{i} \subseteq \setGF_w^{\;\varphi}$ and $\setnu_{i-1} \subseteq \setFG_w^{\,\varphi}$. We prove $\psi_i \in \setFG_w^{\,\varphi}$, i.e., $w \models \F\G \psi_i$, in three steps.

	\begin{itemize}
		\item Claim 1: $\evalnu{\psi_i}{\setmu} = \evalnu{\psi_i}{\setmu_{i}}$.

		By the definition of $\evalnu{\cdot}{\cdot}$, $\evalnu{\psi_i}{\setmu}$ is completely determined by the $\mu$-subformulas of $\psi_i$ that belong to $\setmu$. By the definition of the sequence $(\setmu_0, \setnu_0), \ldots,$ $(\setmu_{n}, \setnu_{n})$, a $\mu$-subformula of $\psi_i$ belongs to $\setmu$ if and only if it belongs to $\setmu_{i}$, and we are done.

		\item Claim 2: $\setmu_i \subseteq \setGF_{w_{k}}^{\;\varphi}$ for every $k \geq 0$.

		Follows immediately from $\setmu_{i} \subseteq \setGF_w^{\;\varphi}$.

		\item Proof of $w \models \F\G \psi_i$.

		By the assumption of (2) we have $w \models \F\G(\evalnu{\psi_i}{\setmu})$, and so, by Claim 1, $w \models \F\G(\evalnu{\psi_i}{\setmu_{i}})$. So there exists an index $j$ such that $w_{j+k} \models \evalnu{\psi_i}{\setmu_i}$ for every $k \geq 0$. By Claim~2 we further have $\setmu_i \subseteq \setGF_{w_{j+k}}^{\;\varphi}$ for every $j, k \geq 0$. So we can apply \Cref{lem:evalnu}.2 to $\setmu_i$, $w_{j+k}$, and $\psi_i$, which yields $w_{j+k} \models \psi_i$ for every $k \geq 0$. So $w \models \F\G \psi_i$.
	\end{itemize}

	\item $\psi_{i} \in \setmu$, i.e., $\setmu_i \setminus \setmu_{i-1} = \{ \psi_i \}$ and $\setnu_i = \setnu_{i-1}$.

	By induction hypothesis we have in this case $\setmu_{i-1} \subseteq \setGF_w^{\;\varphi}$ and $\setnu_{i} \subseteq \setFG_w^{\,\varphi}$. We prove $\psi_i \in \setGF_w^{\;\varphi}$, i.e., $w \models \G\F \psi_i$ in three steps.

	\begin{itemize}
 		\item Claim 1: $\evalmu{\psi_i}{\setnu} = \evalmu{\psi_i}{\setnu_{i}}$.

		The claim is proved as in the previous case.

		\item Claim 2: There is an $j \geq 0$ such that $\setnu_i \subseteq \setG_{w_{k}}^{\varphi}$ for every $k \geq j$.

		Follows immediately from $\setnu_{i} \subseteq \setFG_w^{\,\varphi}$.

 		\item Proof of $w \models \G\F \psi_i$.

		By the assumption of (2) we have $w \models \G\F(\evalmu{\psi_i}{\setnu})$. Let $j$ be the index of Claim~2. By Claim~1 we have $w \models \G\F(\evalmu{\psi_i}{\setnu_{i}})$, and so there exist infinitely many $k \geq j$ such that $w_k \models \evalmu{\psi_i}{\setnu_{i}}$. By Claim~2 we further have $\setnu_i \subseteq \setG_{w_{k}}^{\varphi}$. So we can apply \Cref{lem:evalmu}.2 to $\setnu_i$, $w_{k}$, and $\psi_i$, which yields $w_{k} \models \psi_i$ for infinitely many $k \geq j$. So $w \models \G\F \psi_i$.
	\end{itemize}
\end{itemize}
\end{proof}

\section{Omitted Proofs}
\label{appendix:omitted_proofs}

\lemSizeSf*

\begin{proof}
By \Cref{lem:aww:translation}, some $\alw[2]$ with 
$O(|\subf(\varphi_{\!\setmu,\setnu})|)$ states recognises $\lang(\varphi_{\setmu,\setnu})$. So it suffices to show that $|\subf(\varphi_{\setmu,\setnu}) | \in O(|\subf(\varphi)|)$. This follows from the following claims:

\begin{enumerate}
	\item $|\bigcup \{\subf(\evalnu{\psi}{\setmu}) : \psi \in \subf(\varphi)\} | \leq |\subf(\varphi)|$
	\item $|\bigcup \{\subf(\evalmu{\psi}{\setnu}) : \psi \in \subf(\varphi)\} | \leq |\subf(\varphi)|$
	\item $|\subf(\flatten{\varphi}{\setmu})| \leq 3 |\subf(\varphi)|$
\end{enumerate}

\noindent  Let $\psi_1, \ldots, \psi_n$ be an enumeration of $\subf(\varphi)$ compatible with the subformula order, i.e., if $\psi_i$ is a subformula of $\psi_j$, then $i \leq j$. Let
$X_0 = \emptyset$, and $X_i = X_{i-1} \cup  \{\psi_i\}$ for every $1 \leq i \leq n$
To prove (1-3) we show that for every $0 \leq i \geq n$
\begin{itemize}
\item[(1')] $\left| \bigcup \{\subf(\evalnu{\psi}{\setmu}) : \psi \in X_i\} \right| \leq i$
\item[(2')] $\left| \bigcup \{\subf(\evalmu{\psi}{\setnu}) : \psi \in X_i\} \right| \leq i$
\item[(3')] $\left| \bigcup \{\subf(\flatten{\psi}{M}) \cup \subf(\evalnu{\psi}{\setmu}) : \psi \in X_i\} \right| \leq 3i$
\end{itemize}
\noindent Since $X_n = \subf(\varphi)$, (1) and (2) follow immediately from (1') and (2'), while
(3) follows from (3') and the inclusion 
\[\subf(\flatten{\varphi}{M}) \subseteq \bigcup \{\subf(\flatten{\psi}{M}) \cup \subf(\evalnu{\psi}{\setmu}) : \psi \in \subf(\varphi)\}\]
\noindent which follows easily from the definitions.

We only prove (1') and (3'), since (2') is analogous to (1'). For $i=0$ (1') and (3') hold immediately,  and so it suffices to show
\begin{align*}
       &\left|\bigcup\{\subf(\evalnu{\psi}{\setmu}) : \psi \in X_{i}\}\right| \\
\leq & \left|\bigcup\{\subf(\evalnu{\psi}{\setmu}) : \psi \in X_{i-1}\}\right| + 1 & (*)\\
     & \left|\bigcup\{\subf(\flatten{\psi}{M}) \cup \subf(\evalnu{\psi}{\setmu}) : \psi \in X_{i}\}\right| \\
\leq & \left|\bigcup\{\subf(\flatten{\psi}{M}) \cup \subf(\evalnu{\psi}{\setmu}) : \psi \in X_{i-1}\}\right| + 3 & (**)
\end{align*}
\noindent We prove $(*)$ and $(**)$ by a case distinction on $\psi_i$. We only show one case as an example, since all other cases are either straightforward or analogous.

\smallskip
\noindent Case $\psi_i = \psi_i' \W \psi_i''$. Observe that the subformula ordering ensures $\subf(\psi_i') \subseteq X_{i-1}$ and $\subf(\psi_i'') \subseteq X_{i-1}$. Thus the only \emph{new} proper subformulas we derive are the ones that are directly derived from $\psi_i' \W \psi_i''$. Inserting the definitions for $\subf$, $\evalnu{\cdot}{\cdot}$, and $\flatten{\cdot}{\cdot}$ we obtain the following two set inclusions from which the bound on the cardinality follows:
\begin{align*}
\subf(&\evalnu{\psi_i}{M}) \subseteq  \bigcup\{\subf(\evalnu{\psi}{\setmu}) : \psi \in X_{i-1}\} \\
& \cup \{(\evalnu{\psi_i'}{M}) \W (\evalnu{\psi_i''}{M})\} \\
\subf(&\flatten{\psi_i}{M}) \subseteq  \bigcup\{\subf(\flatten{\psi}{M}) \cup \subf(\evalnu{\psi}{\setmu}) : \psi \in X_{i-1}\} \\
& \cup \{(\flatten{\psi_i'}{M}) \U (\flatten{\psi_i''}{M} \vee \G (\evalnu{\psi_i'}{M})), \G (\evalnu{\psi_i'}{M})\}
\end{align*}
\end{proof}

\propCorrAwwDet*

\begin{proof} Let us sketch the proof.

\medskip

(\aww): The $\subseteq$-inclusion follows immediately from \Cref{lem:det:1,lem:det:3} and \Cref{lem:det:4}. The $\supseteq$-inclusion is a slight adaptation of similar proofs in \cite{DBLP:conf/charme/GurumurthyKSV03}. In order to translate a \dcw{} into a $\aww[2,\R]$ we duplicate the set of states into two sets of marked and unmarked states. We remove from the marked states all rejecting states, and add transitions that allow unmarked states to nondeterministically choose to move to another unmarked state, or to its marked copy. Finally, we define all unmarked states to be rejecting and all marked states to be accepting. The proof of $\aww[2,\A] \supseteq \dbw{}$ is dual. The inclusion $\aww[2] \supseteq \text{$\omega$-regular}$ follows from the previous two results, because every \drw{} is equivalent to a Boolean combination of {\dbw}s and {\dcw}s, which we can express in our initial formula $\theta_0$. The proofs for the remaining inclusions are analogous.

\medskip

(\alw): The $\supseteq$-inclusion for $\Delta_i$ is proven in \Cref{lem:aww:translation}. For a formula $\varphi$ that belongs to $\Sigma_i$ ($\Pi_i$) we also rely on \Cref{lem:aww:translation}, but add a new initial state, $\sta{\varphi}{\Sigma_i}$ ($\sta{\varphi}{\Pi_i}$) that is marked as rejecting (accepting) such that the automaton belongs to $\alw[i,\R]$ ($\alw[i,\A]$). For the $\subseteq$-inclusion, let $\mathcal{A} = \langle \Sigma, Q, \theta_0, \delta, \alpha \rangle$ be a very weak alternating automaton with $\Sigma = 2^{Ap}$. We use the translation from $\alw$ to LTL presented in \cite[Thm. 6]{DBLP:conf/ifipTCS/LodingT00}, with minimal modifications, to define a formula $\chi_\mathcal{A}$ such that $\lang(\chi_\mathcal{A})=\lang(\mathcal{A})$. Then, we show that when $\mathcal{A}$ belongs to one of the classes in the hierarchy,
$\chi_\mathcal{A}$ belongs to the corresponding class of formulas. For the proof of correctness of the translation we refer the reader to \cite{DBLP:conf/ifipTCS/LodingT00}.

For the definition of $\chi_\mathcal{A}$, we assign to every $\theta \in \mathcal{B}^+(Q)$ an LTL formula $\chi(\theta)$ such that $\lang(\chi(\theta)) = \lang(\mathcal{A}_\theta)$, where $\mathcal{A}_\theta$ denotes $\mathcal{A}$ with $\theta$ as initial formula, and set $\chi_\mathcal{A} := \chi(\theta_0)$. Similarly, for the definition of $\chi(\theta)$, we first assign a formula $\chi(q)$ to every state $q$, and then define $\chi(\theta)$ as the result of substituting $\chi(q)$ for $q$ in $\theta$, for every state $q$. It remains to define $\chi(q)$. Using that $\mathcal{A}$ is very weak, we proceed inductively, i.e., we assume that $\chi(q')$ has already been defined for all $q'$ such that $q \rightarrow q'$ and $q \neq q'$. 

For every $q \in Q$ and $\sigma \in 2^{Ap}$, let $\theta_{q,\sigma}$ and $\theta'_{q,\sigma}$ be formulas such that  $\delta(q, \sigma) \equiv (q \wedge \theta_{q,\sigma}) \vee \theta'_{q,\sigma}$ (it is easy to see that they exist). Define
\[\chi(q) = \begin{cases}
	\varphi_q \U \varphi_q' & \text{if $q \notin \alpha$} \\
	\varphi_q \W \varphi_q' & \text{if $q \in \alpha$} \\
\end{cases}\]
\noindent with: \[ \varphi_q  = \bigvee_{\sigma \subseteq \Sigma} \left(\psi_\sigma \wedge \X \chi(\theta_{q,\sigma})\right) \] \[ \varphi_q'  = \bigvee_{\sigma \subseteq \Sigma} \left(\psi_\sigma \wedge \X \chi(\theta'_{q,\sigma})\right)\]  
\[\psi_\sigma  = \bigwedge_{a \in \sigma} a \wedge \bigwedge_{a \notin \sigma} \neg a\]

Since this translation assigns to each $\U$-formula a rejecting state and to each $\W$-formula an accepting state, the syntax tree of $\chi_\mathcal{A}$ has an alternation between $\U$ and $\W$ exactly when there is an alternation between accepting and non-accepting states. This yields all the desired inclusions in $\Sigma_1$, $\Pi_1, \ldots,$ $\Delta_2$. 
\end{proof}

\fi

\end{document}